\documentclass[10pt]{article}
\pdfoutput=1

\usepackage{amssymb}
\usepackage{amsmath}
\usepackage{mathtools}
\usepackage{algpseudocode}
\usepackage{bm}
\usepackage{graphics,epsfig}
\DeclareGraphicsExtensions{.pdf}

\usepackage{url,here}

\usepackage{algorithm}
\allowdisplaybreaks
\usepackage[backref,colorlinks=true]{hyperref}
\usepackage{multirow}

\begin{document}

\newenvironment {proof}{{\noindent\bf Proof.}}{\hfill $\Box$ \medskip}

\newtheorem{theorem}{Theorem}[section]
\newtheorem{lemma}[theorem]{Lemma}
\newtheorem{condition}[theorem]{Condition}
\newtheorem{proposition}[theorem]{Proposition}
\newtheorem{remark}[theorem]{Remark}
\newtheorem{definition}[theorem]{Definition}
\newtheorem{hypothesis}[theorem]{Hypothesis}
\newtheorem{corollary}[theorem]{Corollary}
\newtheorem{example}[theorem]{Example}
\newtheorem{descript}[theorem]{Description}
\newtheorem{assumption}[theorem]{Assumption}
\newcommand{\ag}[1]{{\color{black} #1}}

\def\P{\mathbb{P}}
\def\R{\mathbb{R}}
\def\E{\mathbb{E}}
\def\N{\mathbb{N}}
\def\Z{\mathbb{Z}}

\renewcommand {\theequation}{\arabic{section}.\arabic{equation}}
\def \non{{\nonumber}}
\def \hat{\widehat}
\def \tilde{\widetilde}
\def \bar{\overline}

\def\ind{{\mathchoice {\rm 1\mskip-4mu l} {\rm 1\mskip-4mu l}
{\rm 1\mskip-4.5mu l} {\rm 1\mskip-5mu l}}}
\algnewcommand{\algorithmicgoto}{\textbf{go to step}}%
\algnewcommand{\Goto}[1]{\algorithmicgoto~\ref{#1}}%

\title{\Large\ { \bf A finite state projection method for steady-state sensitivity analysis of stochastic reaction networks}}

\author{Patrik D\"urrenberger, Ankit Gupta, and Mustafa Khammash \\
Department of Biosystems Science and Engineering \\ ETH Zurich \\  Mattenstrasse 26 \\ 4058 Basel, Switzerland. 
}
\date{}
\maketitle
%09-BLAN-0215 
\begin{abstract}
Consider the standard stochastic reaction network model where the dynamics is given by a continuous-time Markov chain over a discrete lattice. For such models, estimation of parameter sensitivities is an important problem, but the existing computational approaches to solve this problem usually require time-consuming Monte Carlo simulations of the reaction dynamics. Therefore these simulation-based approaches can only be expected to work over finite time-intervals, while it is often of interest in applications to examine the sensitivity values at the steady-state after the Markov chain has relaxed to its stationary distribution. The aim of this paper is to present a computational method for the estimation of steady-state parameter sensitivities, which instead of using simulations, relies on the recently developed \emph{stationary Finite State Projection (sFSP)} algorithm [J. Chem. Phys. 147, 154101 (2017)] that provides an accurate estimate of the stationary distribution at a fixed set of parameters. We show that sensitivity values at these parameters can be estimated from the solution of a Poisson equation associated with the infinitesimal generator of the Markov chain. We develop an approach to numerically solve the Poisson equation and this yields an efficient estimator for steady-state parameter sensitivities. We illustrate this method using several examples.
\end{abstract}

%TODO

%TODO
\medskip

\setcounter{equation}{0}

\section{Introduction} \label{sec:intro}
Stochastic models of reaction networks are commonly used in systems and synthetic biology to model reaction dynamics within cells, where some biomolecular reactants are typically present in low copy-numbers \cite{Elowitz}. Such models capture the random timing of reactions and allows one to investigate the role of this randomness in causing cell-to-cell variability and shaping the macroscopic properties of clonal cell-populations \cite{Goutsias}.

In a stochastic model the reaction dynamics is represented as a \emph{continuous-time Markov chain} (CTMC), which keeps track of the molecular or copy-number counts of all the reacting species \cite{DASurvey}. Therefore the CTMC evolves on a discrete state-space $\mathcal{E} \subset \N^d_0$, where $\N_0$ is the set of nonnegative integers and $d$ is the number of reacting species. It is known that the dynamics of the probability distribution of the CTMC is given by the \emph{Chemical Master Equation} (CME) that consists of an ODE describing the inflow and outflow of probabilities at each state in the state-space $\mathcal{E}$. Therefore if this state-space is infinite in size, as is the case in many examples of interest, then the CME is nearly impossible to solve exactly. However approximate solutions to the CME over finite time-periods can be obtained using the \emph{Finite-State Projection} (FSP) that projects the probability dynamics on a finite truncated state-space and solves the resulting system of ODEs \cite{FSP}. Consider the situation when the CTMC is \emph{ergodic} and so the CME solution converges to a unique stationary distribution as time tends to infinity \cite{GuptaPLOS}. In such a setting, the classical FSP is inappropriate for estimating the stationary distribution but recently a modification of this method has been proposed that is able to estimate the stationary distribution accurately under certain conditions \cite{sFSP}. This method is called the \emph{stationary Finite State Projection} (sFSP) and it is described in detail in Section \ref{subsec:sfspmethod}.

In many applications it is of interest to quantitatively determine the influence of some parameter $\theta$ (like temperature, extracellular ligand concentration, cell-volume etc.) on an output of the form $\E( f( X_\theta(t) ) )$, where $\E$ is the expectation operator, $( X_\theta(t) )_{t \geq 0}$ is the $\theta$-dependent CTMC that describes the reaction dynamics, and $f$ is some real-valued function on the CTMC state-space $\mathcal{E}$. This parameter influence is often measured by estimating the infinitesimal sensitivity value   
\begin{align}
\label{defn_paramsens}
S_\theta(f,t) :=  \frac{\partial}{\partial \theta} \E( f( X_\theta(t) ) ). 
\end{align} 
Computing such sensitivities w.r.t. various parameters is useful for many applications, like investigating robustness properties of networks \cite{Stelling}, finding critical reactions, parameter inference \cite{Fink2009} and controlling a system's output \cite{Feng}. Suppose that the CTMC $( X_\theta(t) )_{t \geq 0}$ is ergodic and its unique $\theta$-dependent stationary distribution is $\pi_\theta$. The goal of this paper is to develop a method to numerically estimate the steady-state sensitivity defined by
\begin{align}
\label{defn_ssparamsens}
S_\theta(f):= \lim_{t \to \infty}  \frac{\partial}{\partial \theta} \E( f( X_\theta(t) ) ).
\end{align}
Under certain mild conditions, this limit can be shown to exist and it can be represented as \cite{Gupta3}  
\begin{align}
\label{defn_ssparamsens2}
S_\theta(f) = \frac{ \partial }{ \partial \theta} \left\langle f, \pi_\theta \right\rangle,
\end{align}
where r.h.s.\ is the expectation of function $f$ under the stationary distribution $\pi_\theta$, i.e.
\begin{align}
\label{defn_inner_prod_measure}
 \left\langle f, \pi_\theta \right\rangle =\sum_{x \in \mathcal{E}} f(x) \pi_\theta(x).
\end{align}

Many methods \cite{IRN,Gir,KSR1,KSR2,DA,Our,Gupta2} exist for estimating finite-time sensitivities of the form \eqref{defn_paramsens}, but they rely on simulations of the dynamics $( X_\theta (t) )_{t  \geq 0 }$ obtained with Gillespie's \emph{stochastic simulation algorithm} (SSA) \cite{GP} or its counterparts \cite{NR,AndMod}. As simulations can only be performed over finite time-periods, these methods cannot be naturally extended to estimate steady-state sensitivities \eqref{defn_ssparamsens}. It is of course possible to approximate the steady-state sensitivity $S_\theta(f)$ with the finite-time sensitivity $S_\theta (f,T)$ for a very large time-value $T$. However this approximation is only accurate if the distribution of the random state $X_\theta(T)$ of the dynamics is sufficiently close to the stationary distribution $\pi_\theta$, and it is very difficult to determine how large $T$ needs to be for this to hold. Moreover when $T$ is large, simulations of the dynamics over time-period $[0,T]$ become computationally very expensive, and a large number of such simulations will be required to obtain a statistically useful estimate of $S_\theta (f,T)$. This is because the variance of finite-time sensitivity estimators generally \emph{blows-up} as $T \to \infty$. For the \emph{Likelihood Ratio} (LR) sensitivity estimation method \cite{Gir}, this variance blow-up problem can be circumvented by appropriately \emph{centering} the estimator \cite{wang2018steady}, but such strategies have not been found for other sensitivity estimation methods.

Motivated by the problems faced by simulation-based methods in estimating steady-state sensitivity $S_\theta(f)$, our aim in this paper is to present a simulation-free approach for estimating this sensitivity value. Our method relies on approximate computation of the stationary distribution $\pi_\theta$ with sFSP \cite{sFSP}, and a novel mathematical result which shows how steady-state sensitivity can be evaluated from the solution of a \emph{Poisson Equation} associated with the generator (see Section \ref{subsec:explFormula}) of the CTMC $( X_\theta (t) )_{t  \geq 0 }$. Solving the Poisson Equation is numerically challenging but we develop a \emph{Basis Function Method} (BFM) that is able to efficiently obtain projection of the solution on the linear space spanned by user-specified basis functions (see Section \ref{subsec:Poisson}). We demonstrate that if the collection of basis function is sufficiently large, then BFM can recover the solution of the Poisson Equation \emph{almost exactly}, and this yields a very accurate estimate of the steady-state sensitivity $S_\theta(f)$. We refer to our method as the \emph{Poisson Estimator} (PE) and we illustrate it with a number of examples. We also compare it with the simulation-based  \emph{integrated centered Likelihood Ratio} (IntCLR) \cite{wang2018steady} method that was recently proposed for steady-state sensitivity estimation (see Section \ref{sec:examples}).

We provide a well-documented open source C++ implementation of our method  (see Section \ref{subsec:code}), which includes \emph{parallelized} solvers for the stationary distribution (with sFSP) and for the Poisson Equation (with BFM). We must point out that apart from estimating sensitivities, solving the Poisson Equation for Markov chains has many other applications, for example in computing the optimal policy function for Markov Decision Processes \cite{meyn2008control,feinberg2012handbook}.

\section{Preliminaries} 
\label{sec:prelim}

\subsection{The Stochastic Reaction Network model} \label{subsec:stochasticmodel}

Consider a reaction network involving $M$ biochemical species, denoted by $\mathcal{S}_1,\dots,\mathcal{S}_M$. These species interact via $K$ reactions, and each reaction $k$ has the form
\begin{align}
\label{reactionform}
\sum_{i = 1}^M \nu_{ik} \mathcal{S}_i  \longrightarrow \sum_{i = 1}^M \nu'_{ik} \mathcal{S}_i, 
\end{align}
where $\nu_{ik}$ $(\nu'_{ik})$ is a nonnegative integer specifying the number of $\mathcal{S}_i$ molecules consumed (produced) by reaction $k$. Hence the change in the number of $\mathcal{S}_i$ molecules, caused by the firing of reaction $k$ is given by the integer $\zeta_{i k} = ( \nu'_{ik} -  \nu_{ik} )$ and $\zeta_k = (\zeta_{1 k},\dots, \zeta_{M k}  )$ is called the \emph{stoichiometry vector} for reaction $k$. To each reaction $k$ we associate a \emph{propensity function} $\lambda_k :  \N^M_0 \to [0,\infty)$, that specifies the firing rate of this reaction as $\lambda_k(x_1,\dots,x_M)$, when the molecular count or copy-number of species $\mathcal{S}_i$ is $x_i$. Commonly this rate function is given by \emph{mass-action kinetics} \cite{DASurvey}
\begin{align}
\label{massactionkin_defn}
\lambda_k(x_1,\dots,x_M) = \theta_k \prod_{i=1}^M \ind_{ \{ x_i \geq \nu_{ik} \} } \frac{x_i (x_i-1)\dots (x_i - \nu_{ik} + 1) }{\nu_{ik} !}
\end{align}
where $\theta_k$ is a positive rate constant.

In the classical \emph{continuous-time Markov chain} (CTMC) model of a reaction network \cite{DASurvey}, the state at time $t$ is simply the vector of species copy-numbers $X(t) = ( X_1(t),\dots, X_M(t) ) \in \N^M_0$ at time $t$. When the state is $X(t) = x$, the rate of firing of the $k$-th reaction is $\lambda_k(x)$ and if it fires before any other reaction, the state moves to $(x +\zeta_k)$. In other words, the generator\footnote{A Markov process can be uniquely characterized by its generator, which is an operator specifying the rate of change of the distribution of the Markov process.} of the CTMC $( X(t) )_{t \geq 0}$ is given by
\begin{align}
\label{defn:ctmc_gen}
\mathbb{Q} f(x) = \sum_{k = 1}^K \lambda_k(x) (f(x +\zeta_k)  - f(x)),
\end{align}
where $f$ is any bounded real-valued function on the state-space $\mathcal{E} \subset \N^M_0$ of the CTMC. We shall assume that this state-space $\mathcal{E}$ is non-empty and \emph{closed} under the CTMC dynamics, i.e. if $x$ is a state in $\mathcal{E}$ and $k$ is some reaction with $ \lambda_k(x) > 0$ then $(x+\zeta_k)$ is also a state in $\mathcal{E}$. 

As the state-space $\mathcal{E}$ is countable, we can construct a \emph{bijection} $\phi : \mathcal{E} \longrightarrow \{0,1\dots, | \mathcal{E} |\} $, where $| \mathcal{E} |$ denotes the number of elements in $\mathcal{E}$. Letting $x_i = \phi^{-1}(i)$ for $i=0,1,\dots$, we can represent the state-space as $\mathcal{E} = \{x_0,x_1,\dots\}$, and express the operator $\mathbb{Q}$ \eqref{defn:ctmc_gen} as the transition rate matrix $Q = [Q_{ij}]$
\begin{align}
\label{defn_transitionratematrix}
Q_{ij} = \left\{  
\begin{array}{cc}
- \sum_{k=1}^K \lambda_k(x_i) & \textnormal{ if } i =j \\
  \sum_{k \in \mathcal{K}_{ij}} \lambda_k(x_i) & \textnormal{ if } i \neq j  
\end{array} \right. 
\end{align}
where $\mathcal{K}_{ij} = \{ k =1\dots,K : x_j = x_i +\zeta_k\}$ is the set of reactions that can take the state from $x_i$ to $x_j$ in a single firing. Note that if $\mathcal{K}_{ij} = \emptyset$ then $Q_{ij} = 0$. Moreover the matrix $Q$ is bi-infinite when $| \mathcal{E} | =\infty$, which is often the case in examples of interest.

Suppose $( X(t) )_{t \geq 0}$ is the CTMC with transition rate matrix $Q$ and some initial state $X(0) \in \mathcal{E}$. Let $p(t) = (p_0(t),p_1(t),\dots)$ where
\begin{align*}
p_i(t) = \P( X(t) = x_i )
\end{align*}
is the probability that the dynamics is in state $x_i$ at time $t$. The time-evolution of the probability vector is given by the well-known \emph{Chemical Master Equation} (CME) which can be expressed as
\begin{align}
\label{defn_cme}
\frac{dp}{dt} = Q^T p(t).
\end{align}
Even though the CME is a first-order linear system of ODEs, solving it analytically is infeasible when the state-space $\mathcal{E}$ is large or infinite in size. For this reason the Finite State Projection (FSP) \cite{FSP} was developed to approximately solve the CME by projecting it on a truncated state-space $\mathcal{E}_n \subset \mathcal{E}$.

\subsection{Ergodicity of the CTMC} \label{sec:ctmcergodic}

Often one is interested in the steady-state behavior of the CTMC $( X(t) )_{t  \geq 0}$, which is characterized by a stationary distribution $\pi = (\pi_0,\pi_1,\dots)$ which acts as a \emph{fixed-point} for the CME \eqref{defn_cme},
\begin{align}
\label{fixedpteqn}
Q^T \pi = {\bf 0}. 
\end{align}
%0 is a vector of all zeroes...
The CTMC is called \emph{ergodic} if the fixed point is unique and globally attracting in the sense that
\begin{align*}
\lim_{t \to \infty} \| p(t) - \pi \|_{ \ell_1}:= \sum_{i = 0}^{ | \mathcal{E} |} |  p_i(t) - \pi_i  | =0.
\end{align*}
In this paper we work under the assumption of \emph{exponential ergodicity} \cite{Meyn} which holds if the above convergence is exponentially fast, i.e. there exists a constant $\rho > 0$ and another constant $C > 0$, that may depend on the initial distribution $p(0)$, such that for any $t > 0$ 
\begin{align}
\label{defn_experg}
 \| p(t) - \pi \|_{ \ell_1} \leq C e^{ - \rho t }.
\end{align}

Verifying ergodicity of the CTMC corresponding to a reaction network is a challenging problem when the state-space is infinite. One way to address this challenge is to first establish that the underlying state-space $\mathcal{E}$ is \emph{irreducible}, i.e. all states in $\mathcal{E}$ are accessible from each other via a sequence of positive-propensity reactions, and then construct a Foster-Lyapunov function on the state-space to show that the dynamics has a tendency to be attracted to a compact set within the state-space \cite{Meyn}. In particular, using Theorem 7.1 in \cite{Meyn} we can prove exponential ergodicity of the CTMC by showing that there exists a \emph{norm-like} function\footnote{A positive function is called norm-like if all its sub-level sets are compact.} $V : \mathcal{E} \to [1,\infty)$ such that for some $C_1,C_2 > 0$, we have
\begin{align}
\label{normlikecondlyapcond}
\mathbb{Q} V(x) \leq C_1 - C_2 V(x) \qquad \textnormal{for all} \qquad x \in \mathcal{E},
\end{align}
where $\mathbb{Q}$ \eqref{defn:ctmc_gen} is the generator of the CTMC. In \cite{GuptaPLOS} and \cite{gupta2018computational} computational procedures have been developed that systematically check state-space irreducibility and construct the required Foster-Lyapunov function $V$ satisfying \eqref{normlikecondlyapcond} under some conditions. In fact, for many systems biology networks, a linear Foster-Lyapunov function of the form
\begin{align}
\label{linearlyapunovfunction}
V(x) = 1 + \langle v,x\rangle
\end{align}
can be constructed \cite{GuptaPLOS}. Here $\langle \cdot,\cdot \rangle$ denotes the standard inner product on $\R^d$ and $v \in \R^d$ is a component-wise positive vector that can be found via linear programming. Note that for this linear Foster-Lyapunov function, condition \eqref{normlikecondlyapcond} can be equivalently expressed as
\begin{align}
\label{normlikecondlyapcond2}
\sum_{k=1}^K \lambda_k(x) \langle v,\zeta_k \rangle \leq (C_1 - C_2) - C_2  \langle v,x\rangle  \qquad \textnormal{for all} \qquad x \in \mathcal{E}.
\end{align}
Often for typical reaction networks, in addition to this condition we also have that for some $C_3,C_4 > 0$
 \begin{align}
\label{normlikecondlyapcond3}
\sum_{k=1}^K \lambda_k(x) \langle v,\zeta_k \rangle^2  \leq C_3 + C_4  \langle v,x\rangle  \qquad \textnormal{for all} \qquad x \in \mathcal{E},
\end{align}
which ensures that all statistical moments of the stationary distribution $\pi$ are finite (see Theorem 5 in \cite{GuptaPLOS}).

\subsection{The Stationary Finite State Projection Algorithm} \label{subsec:sfspmethod} 

The classical FSP algorithm \cite{FSP} introduces an \emph{absorbing} state into the dynamics, in order to restrict it to a finite truncated subset. Therefore even though FSP is very successful in solving the CME over finite time-intervals, it cannot be used to estimate the stationary distribution $\pi$, as all the probability-mass will flow into the absorbing state at steady-state (see \cite{sFSP} for more details).

\begin{figure}[h]
	\label{sFSP_idea}
		\includegraphics[width=0.8\textwidth]{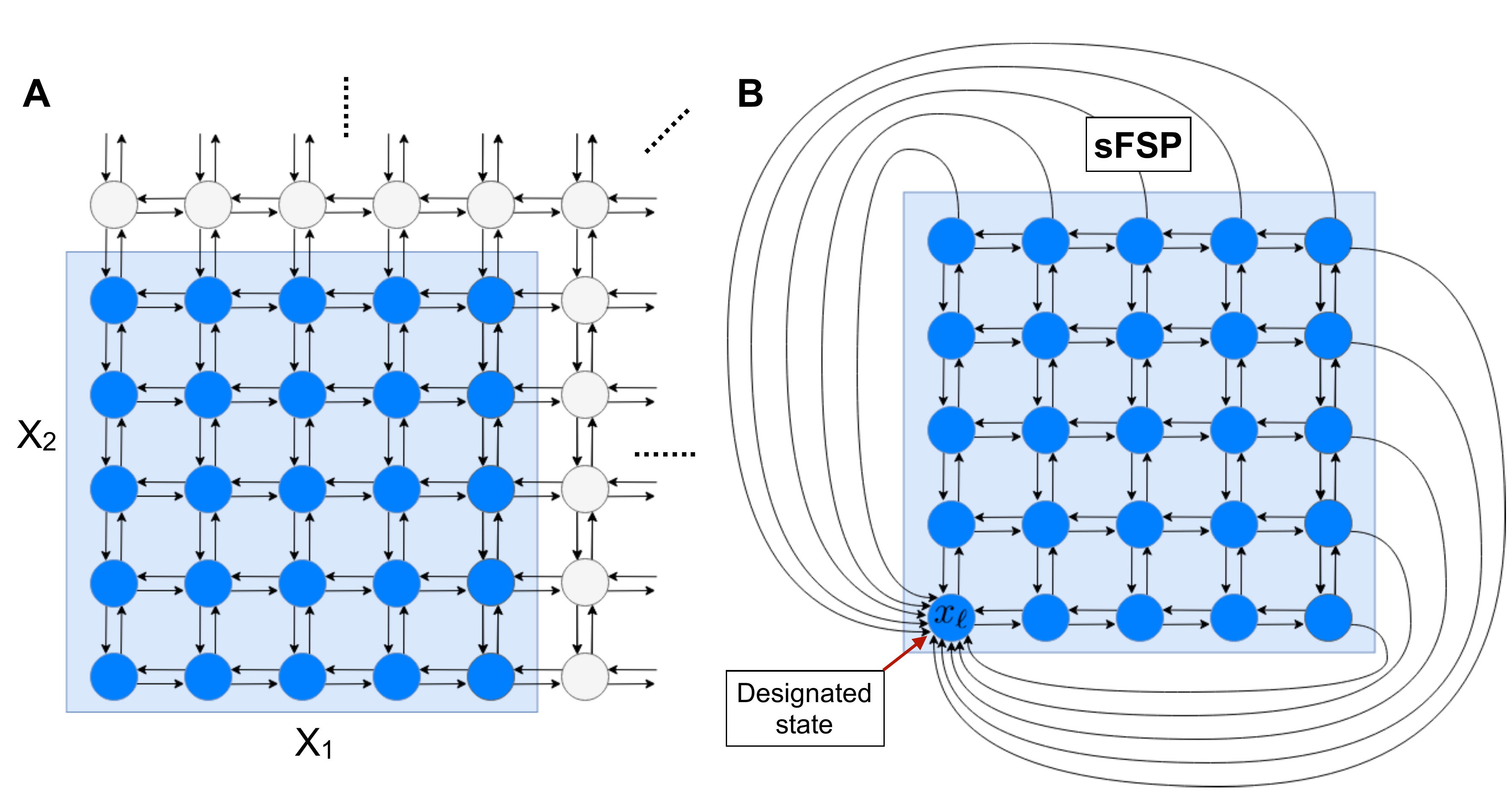}
	\caption{In (A) we depict the infinite state-space $\mathcal{E} = \N^2_0$ of a two-species network. The shaded area represents a finite truncated state-space, on which the dynamics is projected in the \emph{stationary Finite State Projection} (sFSP) method, by directing all the outgoing transitions to some designated state $x_\ell$ (see (B)). }
\end{figure}

Motivated by this problem, recently the \emph{stationary Finite State Projection} (sFSP) algorithm \cite{sFSP} was proposed, that instead of using an absorbing state to capture all the outgoing transitions from the truncated state-space, forces all the outgoing transitions to flow into a \emph{designated} state within the truncated state-space (see Figure \ref{sFSP_idea}). More formally, suppose $\mathcal{E}_n =\{ x_{j_1} ,\dots, x_{j_n}\} \subset \mathcal{E}$ is the truncated state-space of size $n$, and let $Q_n$ be the matrix formed by considering only the rows and columns of $Q$ that are indexed by $j_1,\dots,j_n$. Furthermore, let $c_n $ be the $n$-dimensional nonnegative column vector whose $i$-th component is the sum of the rates of all the reactions that originate at state $x_i$ and terminate at a state outside $\mathcal{E}_n$,
\begin{align*}
c_{n,i} = \sum_{k=1}^K \ind_{ \{ (x_{j_i} +\zeta_k) \notin \mathcal{E}_n \} } \lambda_k(x_{j_i}). 
\end{align*}
Assuming that the $\ell$-th state $x_{j_\ell}$ is the designated state, in the sFSP algorithm the CTMC dynamics on the truncated state-space $\mathcal{E}_n$ has the transition rate matrix given by
\begin{align}
\label{transratematrixsfsp}
\bar{Q}_n = Q_n + c_n e^T_\ell,
\end{align}
where $e_\ell$ is the $n \times 1$ vector whose $\ell$-th component is $1$ and the rest are zeros. The stationary distribution $\pi$ is then estimated by computing the finite-dimensional stationary distribution $\pi_n$ on $\mathcal{E}_n$ for the transition rate matrix $\bar{Q}_n$. This can be done by solving the linear system \eqref{fixedpteqn} with $Q$ and $\pi$ replaced by $\bar{Q}_n$ and $\pi_n$ respectively.

We now make some assumptions which are required to certify the accuracy of sFSP. 
\begin{assumption}
\label{assmp:main}
\item[(A)] The state-space $\mathcal{E}$ is irreducible for the original CTMC with transition rate matrix $Q$.
\item[(B)] There exists a Foster-Lyapunov function $V : \mathcal{E} \to [1,\infty)$ satisfying \eqref{normlikecondlyapcond}.
\item[(C)] $\{ \mathcal{E}_n : n=1,2,\dots \}$ is a sequence of increasing finite sets (i.e. $n \leq m$ implies $\mathcal{E}_n \subset \mathcal{E}_m$) which converge to the full state-space $\mathcal{E}$ in the limit $n\to\infty$.
\end{assumption}
Under these assumptions it can be shown (see Theorem 3.1 in \cite{sFSP}) that the stationary distribution $\pi_n$ for the projected CTMC on truncated state-space $\mathcal{E}_n$ exists uniquely and it converges to the stationary distribution $\pi$ for the original CTMC in the $\ell_1$ metric
\begin{align}
\label{limitresult}
\lim_{n \to \infty} \| \pi -  \pi_n \|_{ \ell_1} = 0.
\end{align}
This shows that for a large enough truncated state-space $\mathcal{E}_n$, the sFSP-estimated stationary distribution $\pi_n$ is likely to be ``close" to the true stationary distribution $\pi$. More precise estimates on the approximation error can be found in \cite{sFSP} and it involves the \emph{outflow} rate, measured as $\langle c_n, \pi_n \rangle$, at the estimated stationary distribution.

\section{Steady-State Sensitivity Estimation} \label{sec:sensEst}

Consider the stochastic model of a reaction network, as described in Section \ref{subsec:stochasticmodel}, and suppose that in addition to the state $x$ the propensity functions $\lambda_k$-s depend on a scalar parameter $\theta$ and the mapping $\theta \mapsto \lambda_k(x,\theta)$ is differentiable for each fixed $x$. Replacing each $\lambda_k(x)$ by $\lambda_k(x,\theta)$ in \eqref{defn:ctmc_gen}, we obtain the generator $\mathbb{Q}_\theta$ of the $\theta$-dependent CTMC $( X_\theta(t) )_{ t \geq 0 }$. 

In this section, we present the main contributions of this paper, which is developing a method to estimate the steady-state sensitivity $S_\theta(f)$ defined by the limit \eqref{defn_ssparamsens}. In the next section we prove that this limit exists under certain assumptions, and connect it with the solution of a Poisson Equation for the generator $\mathbb{Q}_\theta$.

\subsection{An explicit expression for steady-state sensitivity} \label{subsec:explFormula}

Suppose parts (A) and (B) of Assumption \ref{assmp:main} hold for the original CTMC with generator $\mathbb{Q}_\theta$, and hence the state-space $\mathcal{E}$ is irreducible and suppose $V$ is the Foster-Lyapunov function satisfying \eqref{normlikecondlyapcond}. Under these assumptions this CTMC is exponentially ergodic with a unique stationary distribution $\pi_\theta$. Let $f : \mathcal{E} \to \R$ be the output function for which the steady-state sensitivity $S_\theta(f)$ needs to be computed. The Poisson equation for the pair $( \mathbb{Q}_\theta,f )$ is given by
\begin{align}
\label{defn_poiss_eqn1}
- \mathbb{Q}_\theta g = f - \langle f,\pi_\theta \rangle,
\end{align}
where the inner product $\langle f,\pi_\theta \rangle$ is defined as in \eqref{defn_inner_prod_measure}. It is known that the solution $g$ of this Poisson equation, if it exists, is unique up to addition by a constant function.

We next present our main result in this paper which connects the steady-state sensitivity $S_\theta(f)$ to the solution of the Poisson equation \eqref{defn_poiss_eqn1}. Henceforth, we call a function $f : \mathcal{E}  \to \R$ \emph{polynomially growing}, if there exist constants $C,r>0$ such that
\begin{align*}
| f(x)  | \leq C(  1 + \| x\|^r) \qquad \textnormal{for each} \qquad x \in \mathcal{E},
\end{align*}
where $\|\cdot\|$ denotes the standard norm on $\R^d$.
\begin{theorem}
\label{theorem:main}
Suppose parts (A) and (B) of Assumption \ref{assmp:main} hold for the CTMC with generator $\mathbb{Q}_\theta$, and with a linear Foster-Lyapunov function $V$ \eqref{linearlyapunovfunction} satisfying \eqref{normlikecondlyapcond2} and \eqref{normlikecondlyapcond3}. Furthermore assume that the output function $f$ and $\theta$-derivatives of the propensity functions (i.e. $\partial \lambda_k/\partial \theta$ for $k=1,\dots,K$) are polynomially growing. Then there exists a solution $g$ of the Poisson equation \eqref{defn_poiss_eqn1} and the steady-state sensitivity is given by
\begin{align}
\label{exact_sens_formula}
S_\theta(f) = \sum_{k=1}^K \sum_{x \in \mathcal{E}}  \frac{ \partial \lambda_k(x ,\theta) }{ \partial \theta } (g(x+\zeta_k) -g(x))  \pi_\theta(x)
\end{align}
where $\pi_\theta$ is the stationary distribution of the CTMC.
\end{theorem}
\begin{remark} \label{remark1}
Note that the sensitivity value $S_\theta(f) $ is independent of the choice of the solution $g$ of the Poisson equation \eqref{defn_poiss_eqn1}, as all such solutions differ by a constant function and hence the differences $g(x+\zeta_k) -g(x)$ are the same for all the solutions.
\end{remark}
\begin{proof}
Let $V$ be the linear Foster-Lyapunov function \eqref{linearlyapunovfunction} satisfying \eqref{normlikecondlyapcond2} and \eqref{normlikecondlyapcond3}. Corresponding to $V$ we can define the $V$-norm of any function $h:\mathcal{E} \to \R$ as 
\begin{align}
\label{defn_v_norm}
\| h\|_V = \sup_{x \in \mathcal{E}} \frac{ | h(x) | }{V(x)}.
\end{align}
It is known that if function $f$ has finite $V$-norm (i.e. $\|f\|_V < \infty$), then there exists a solution $g$ of the Poisson equation \eqref{defn_poiss_eqn1} which also has finite $V$-norm (see Theorem 2.3\footnote{This result has been proved for discrete-time Markov chains but it can be easily transformed to the continuous-time setting using the resolvent operator as mentioned in the proof of Proposition \ref{proposition:approx}.} in \cite{glynn1996liapounov}).

In our case the output function $f$ is polynomially growing and so $\|f\|_V$ may not be finite. However for any positive integer $m$, if we define a function $V_m : \mathcal{E} \to [1,\infty)$ by
\begin{align}
\label{genfosterlyapfunction}
V_m(x) = 1 + \langle v,x \rangle^m,
\end{align}
then it can be seen from the proof of Theorem 5 in \cite{GuptaPLOS} that, under conditions \eqref{normlikecondlyapcond2} and \eqref{normlikecondlyapcond3}, $V_m$ also serves as a Foster-Lyapunov function for the dynamics (i.e.\ it satisfies \eqref{normlikecondlyapcond} for generator $\mathbb{Q}_\theta$). Moreover from Theorem 4.2 in \cite{Meyn} it can be concluded that the stationary distribution $\pi_\theta$ is such that
\begin{align}
\label{stationarydistributionnbound}
\| \pi_\theta \|_{ V_m}:= \sum_{x \in \mathcal{E}} V_m(x) \pi_\theta(x) < \infty.
\end{align}
As our function $f$ is polynomially growing, we have $\| f \|_{ V_m} < \infty$ for some positive integer $m$, and hence there is a solution $g$ of the Poisson equation \eqref{defn_poiss_eqn1} which also satisfies $\| g \|_{ V_m} < \infty$. Similarly as $\theta$-derivatives of the propensity functions are polynomially growing as well, they also have finite $V_m$-norm for some $m$. As \eqref{stationarydistributionnbound} holds for any $m$, it can be seen that the r.h.s.\ of \eqref{exact_sens_formula} is finite even though the state-space $\mathcal{E}$ may be countably infinite.

We now prove that the steady-state sensitivity $S_\theta(f)$, defined by \eqref{defn_ssparamsens}, is equal to the r.h.s.\ of \eqref{exact_sens_formula}. Let $( Y_\theta(t) )_{t \geq 0}$ be a CTMC with generator $\mathbb{Q}_\theta$ and for any state $y \in \mathcal{E}$ and polynomially growing function $h:\mathcal{E} \to \R$ let
\begin{align*}
\Psi(y,h,t) = \E\left( h( Y_\theta(t) ) \vert Y_\theta(0) =y \right).
\end{align*}
Under conditions \eqref{normlikecondlyapcond2} and \eqref{normlikecondlyapcond3}, this expectation is finite. Furthermore the limit
\begin{align*}
\lim_{t \to \infty}  \Psi(y,h,t) = \langle h,\pi_\theta \rangle
\end{align*}
holds (see Proposition S2.2 in \cite{GuptaPLOS}), where the inner product between function $h$ and stationary distribution $\pi_\theta$ is defined as in \eqref{defn_inner_prod_measure}. Moreover since $h$ is polynomially growing, there exist constants $c,\beta > 0$ and positive integer $m$ such that 
\begin{align}
\label{fexponentialergodicitycriterion}
| \Psi(y,h,t) - \langle h,\pi_\theta \rangle | \leq  \beta V_m(y) e^{ -c t} \quad \textnormal{for all} \quad t \geq 0 \quad \textnormal{and} \quad y \in \mathcal{E},
\end{align}
where $V_m$ is the Foster-Lyapunov function given by \eqref{genfosterlyapfunction}. This relation follows from Theorem 6.1 in \cite{Meyn}.

Let $S_\theta(f,t)$ be the finite-time sensitivity value defined by \eqref{defn_paramsens} for the CTMC $( X_\theta(t) )_{t \geq 0}$ with generator $\mathbb{Q}_\theta$ and initial state $x_0$. From Theorem 3.3 in \cite{gupta2018estimation}, we can express $S_\theta(f,t)$ as
\begin{align}
\label{finitesensitivityexpression}
S_\theta(f,t) = \sum_{k=1}^K \E\left[ \int_0^t \frac{ \partial \lambda_k( X_\theta(s) ,\theta ) }{ \partial \theta } \left(  \Psi(X_\theta(s) +\zeta_k ,f,t - s) - \Psi(X_\theta(s),f,t - s)    \right) ds  \right].
\end{align}
Using \eqref{fexponentialergodicitycriterion} we know that for some constants $c,\beta > 0$ and $m$
\begin{align}
\label{differenceofpsis}
&|  \Psi(X_\theta(s) +\zeta_k ,f,t - s) - \Psi(X_\theta(s),f,t - s) | \notag \\
& \leq |  \Psi(X_\theta(s) +\zeta_k ,f,t - s) -  \langle f,\pi_\theta \rangle | + |  \Psi(X_\theta(s) ,f,t - s) -  \langle f,\pi_\theta \rangle |  \notag \\
& \leq \beta( V_m(X_\theta(s) +\zeta_k) +V_m(X_\theta(s) )  ) e^{ -c(t-s)}.
\end{align}
This inequality implies that the integrand in \eqref{finitesensitivityexpression} can be bounded by the product of a polynomially growing function of the state $X_\theta(s)$ and an exponential decay function. Since the moments of the CTMC $( X_\theta(t) )_{t \geq 0}$ can be uniformly bounded over the infinite time-interval $[0,\infty)$ (see Theorem 2 in \cite{GuptaPLOS}), the integral on the r.h.s.\ of \eqref{differenceofpsis} exists as $t \to \infty$ and so we can write the steady-state sensitivity as
\begin{align*}
S_\theta(f) &= \lim_{t \to \infty} S_\theta(f,t) \\
&=  \lim_{t \to \infty} \sum_{k=1}^K \E\left[ \int_0^t \frac{ \partial \lambda_k( X_\theta(s) ,\theta ) }{ \partial \theta } \left(  \Psi(X_\theta(s) +\zeta_k ,f,t - s) - \Psi(X_\theta(s),f,t - s)    \right) ds  \right].
\end{align*}
Notice that due to \eqref{differenceofpsis} and uniform moment boundedness, for any $u\geq 0$ we have
 \begin{align*}
 \lim_{t \to \infty} \sum_{k=1}^K \E\left[ \int_0^u \frac{ \partial \lambda_k( X_\theta(s) ,\theta ) }{ \partial \theta } \left(  \Psi(X_\theta(s) +\zeta_k ,f,t - s) - \Psi(X_\theta(s),f,t - s)    \right) ds  \right] = 0,
\end{align*}
and therefore by exploiting the linearity of the integration and expectation operators we can represent the steady-state sensitivity as
\begin{align}
\label{ss_sensitvitiy_form2}
&S_\theta(f)  \\
&=   \lim_{u \to \infty} \lim_{t \to \infty} \sum_{k=1}^K \E\left[ \int_u^t \frac{ \partial \lambda_k( X_\theta(s) ,\theta ) }{ \partial \theta } \left(  \Psi(X_\theta(s) +\zeta_k ,f,t - s) - \Psi(X_\theta(s),f,t - s)    \right) ds  \right].\notag
\end{align}
Letting $p_\theta(s,x) = \P( X_\theta(s) =x )$ we obtain
\begin{align*}
S_\theta(f) & =  \lim_{u \to \infty} \lim_{t \to \infty} \sum_{k=1}^K  \sum_{x \in \mathcal{E} } \frac{ \partial \lambda_k( x,\theta ) }{ \partial \theta }  \left[ \int_u^t \left(  \Psi(x +\zeta_k ,f,t - s) - \Psi(x,f,t - s)    \right) p_\theta(s,x) ds  \right] \\
& =  \lim_{u \to \infty} \lim_{t \to \infty} \sum_{k=1}^K  \sum_{x \in \mathcal{E} } \frac{ \partial \lambda_k( x,\theta ) }{ \partial \theta }  \left[ \int_0^{t-u} \left(  \Psi(x +\zeta_k ,f,s) - \Psi(x,f,s)    \right) p_\theta(t-s,x) ds  \right] \\
& =    \sum_{k=1}^K  \sum_{x \in \mathcal{E} } \frac{ \partial \lambda_k( x,\theta ) }{ \partial \theta }  \left[ \int_0^{\infty} \left(  \Psi(x +\zeta_k ,f,s) - \Psi(x,f,s) \right)  ds    \right] \pi_\theta(x) 
\end{align*}
where the last relation follows because \eqref{fexponentialergodicitycriterion} holds for any $m$, and $(t-s) \geq u$ for $s \in [0,t-u]$ and hence if $u$ is large $p_\theta(t-s,x)$ is close to the stationary distribution $\pi_\theta$ in the norm $\| \cdot\|_{V_m}$ defined by \eqref{stationarydistributionnbound}.

Now to prove this theorem it suffices to show that
\begin{align}
\label{resolventidentity}
 \left[ \int_0^{\infty} \left(  \Psi(x +\zeta_k ,f,s) - \Psi(x,f,s) \right)  ds    \right] = g(x+\zeta_k) -g(x) \quad \textnormal{for each} \quad x \in \mathcal{E}
\end{align}
where $g$ is a solution of the Poisson equation \eqref{defn_poiss_eqn1}. Let $( Y_1(t) )_{t \geq 0 }$ and $( Y_2(t) )_{t \geq 0 }$ be two CTMCs with generator $\mathbb{Q}_\theta$ and initial states $x$ and $(x+\zeta_k)$ respectively. Then applying Dynkin’s formula (see \cite{EK}) on function $g$ we obtain
\begin{align}
\label{dynkineqn1}
\E(g( Y_1(t) ) ) &= g(x) + \E \left [ \int_{0}^t \mathbb{Q}_\theta g( Y_1(s) )ds  \right] \notag \\
&= g(x) + \langle f, \pi_\theta  \rangle t -  \E \left [ \int_{0}^t f( Y_1(s) )ds  \right] \notag \\
&= g(x) + \langle f, \pi_\theta  \rangle t -  \int_{0}^t \Psi(x,f,s) ds   
\end{align}
and
\begin{align}
\label{dynkineqn2}
\E(g( Y_2(t) ) ) &= g(x+\zeta_k) + \E \left [ \int_{0}^t \mathbb{Q}_\theta g( Y_2(s) )ds  \right]  \notag \\
&= g(x+\zeta_k) + \langle f, \pi_\theta  \rangle t -  \E \left [ \int_{0}^t f( Y_2(s) )ds  \right]   \notag \\
& = g(x+\zeta_k) + \langle f, \pi_\theta  \rangle t -   \int_{0}^t \Psi(x+\zeta_k,f,s) ds.  
\end{align}
Subtracting \eqref{dynkineqn1} from \eqref{dynkineqn2} and rearranging we get
\begin{align}
\label{dynkineqn3}
g(x +\zeta_k) - g(x) = \int_{0}^t \left( \Psi(x+\zeta_k,f,s) -  \Psi(x,f,s) \right) ds + \E(g( Y_2(t) ) ) -\E(g( Y_1(t) ) ). 
\end{align}
Due to ergodicity $\lim_{t \to \infty} \E(g( Y_2(t) ) ) = \lim_{t \to \infty} \E(g( Y_1(t) ) ) = \langle g, \pi_\theta \rangle$, and hence letting $t \to \infty$ in \eqref{dynkineqn3} proves \eqref{resolventidentity} and completes the proof of this theorem.
\end{proof}

\subsection{Approximation of steady-state sensitivity} \label{subsec:sensapprox}

We mentioned before, that in most examples of interest the state-space $\mathcal{E}$ is infinite in size. Hence the stationary distribution $\pi_\theta$ cannot be exactly computed, and the problem of solving the Poisson equation \eqref{defn_poiss_eqn1} becomes infinite-dimensional. Due to these issues the exact formula \eqref{exact_sens_formula} cannot be directly used to estimate the steady-state sensitivity value, and we need to replace the stationary distribution $\pi_\theta$ and the Poisson equation solution $g$ with its approximations. Let $\mathcal{E}_n = \{x_{j_1},\dots, x_{j_n}\}$ be the truncated state-space as in part (C) of Assumption \ref{assmp:main}, and let $\pi_{n,\theta}$ be the sFSP estimated approximation (see Section \ref{subsec:sfspmethod}) of $\pi_\theta$ with $x_{j_\ell}$ as the designated state. We can view $\pi_{n,\theta}$ as a column-vector which sums to one and is in the left null-space of the matrix $\bar{Q}_{n,\theta}$ defined by \eqref{transratematrixsfsp}.

We define the finite-dimensional Poisson equation, analogous to \eqref{defn_poiss_eqn1} as
\begin{align}
\label{defn_poiss_eqn2}
- \bar{Q}_{n,\theta}  g_n = f_n - \langle f_n, \pi_{n,\theta} \rangle {\bf 1},
\end{align}
where ${\bf 1}$ is the $ n \times 1$ vector of all ones, and $f_n$ denotes the $ n \times 1$ vector given by $f_n = (f(x_{j_1}),\dots,  f(x_{j_n}))$. A solution $g_n$ of this Poisson equation exists because $\pi_{n,\theta}$ is orthogonal to the r.h.s. of \eqref{defn_poiss_eqn2} and it spans the one-dimensional left null-space of matrix $\bar{Q}_{n,\theta} $. This null-space is one-dimensional because the CTMC specified by transition-rate matrix $\bar{Q}_{n,\theta}$ has a unique stationary distribution over the state-space $\mathcal{E}_n$. This uniqueness follows from the fact that the construction of $\bar{Q}_{n,\theta} $ ensures that the truncated state-space $\mathcal{E}_n$ is irreducible for the corresponding CTMC dynamics (for details see Section III.A in \cite{sFSP}). As for the original Poisson equation \eqref{defn_poiss_eqn1}, a solution $g_n$ of the finite-dimensional Poisson equation \eqref{defn_poiss_eqn2}, is unique up to addition by a constant vector. Using a solution $g_n$ we define the approximate steady-state sensitivity as     
\begin{align}
\label{approx_sens_formula}
S_{n,\theta}(f) =  \sum_{k=1}^K \sum_{x \in \mathcal{E}_n}  \frac{ \partial \lambda_k(x ,\theta) }{ \partial \theta } (g_n(x+\zeta_k) -g_n(x))  \pi_{n,\theta}(x),
\end{align}
where we set $(x +\zeta_k)$ to the designated state $x_{j_\ell}$ whenever state $x$ is in the boundary of $\mathcal{E}_n$ denoted by 
\begin{align*}
\textnormal{Bdry}( \mathcal{E}_n ) = \{ x \in \mathcal{E}_n :  \lambda_k(x) > 0 \textnormal{ and } (x +\zeta_k) \notin \mathcal{E}_n \textnormal{ for some } k=1,\dots,K \}.
\end{align*}
Note that in this formula we regard $g_n$ and $\pi_{n,\theta}$ as real-valued functions over $\mathcal{E}_n$. The next proposition shows that for a large enough truncated state-space $\mathcal{E}_n$, $S_{n,\theta}(f) $ is an accurate approximation of the true steady-state sensitivity value $S_{\theta}(f)$.

 \begin{proposition}
\label{proposition:approx}
Suppose part (C) of Assumption \ref{assmp:main} holds along with all the conditions of Theorem \ref{theorem:main}. Furthermore assume that the sequence of truncated state-spaces $\{ \mathcal{E}_n \}$ grows uniformly w.r.t.\ to the Foster-Lyapunov function $V$ in the sense that there exists a constant $\theta \in (0,1)$ such that
\begin{align*}
\min_{ x \in \textnormal{Bdry}( \mathcal{E}_n ) } V(x) \geq \theta \max_{ x \in \textnormal{Bdry}( \mathcal{E}_n ) } V(x) \quad \textnormal{for all} \quad n=1,2,\dots.
\end{align*}
Then letting $S_{n,\theta}(f)$ be defined by \eqref{approx_sens_formula}, we have 
\begin{align}
\label{sensconvrelationmain}
\lim_{n \to \infty}  S_{n,\theta}(f) =  S_{\theta}(f).
\end{align}
\end{proposition}

\begin{proof}
In this proof we drop the subscript $\theta$ for convenience. Let $Q$ be the bi-infinite transition rate matrix (see \eqref{defn_transitionratematrix}) of the original CTMC with stationary distribution $\pi$ over state-space $\mathcal{E} = \{x_0,x_1,\dots\}$. Similarly let $\bar{Q}_n$ be the finite transition rate matrix (see \eqref{transratematrixsfsp}) of the CTMC projected over the truncated state-space $\mathcal{E}_n$ which has stationary distribution $\pi_n$. Without loss of generality we can assume that the truncated state-space comprises the first $n$ states in $\mathcal{E}$, i.e.\ $\mathcal{E} = \{x_0,\dots,x_{n-1}\}$, and the designated state is $x_\ell =x_0$. Henceforth let $g$ and $g_n$ be the unique solutions of the Poisson equations \eqref{defn_poiss_eqn1} and \eqref{defn_poiss_eqn2} respectively, with $\langle g,\pi\rangle  = \langle g_n, \pi_n \rangle = 0$ (see Remark \ref{remark1}).

As explained in the proof of Theorem \ref{theorem:main}, for any positive integer $m$, the function $V_m$ (see \eqref{genfosterlyapfunction}) serves as a Foster-Lyapunov function for the original CTMC dynamics. The same holds true for the projected CTMC dynamics with transition rate matrix $\bar{Q}_n$, and the constants $C_1$ and $C_2$ that appear in the corresponding Foster-Lyapunov condition \eqref{normlikecondlyapcond} are independent of $n$ (see \cite{Hart}). Appealing to this fact and Theorem 2.3 in \cite{glynn1996liapounov}, we can conclude that the family of solutions $\{g_n\}$ to the Poisson equation \eqref{defn_poiss_eqn2} is uniformly polynomially growing, i.e.\ there exist positive constants $C$ and $m$ such that for each $n$ and $x \in \mathcal{E}$
\begin{align}
\label{uniformgnolybound}
| g_n(x)  | \leq C(  1 + \| x\|^m). 
\end{align}
Moreover from Theorem 3.1 in \cite{sFSP} we know that for any $m$
\begin{align}
\label{stationarydistributionconvmnorm}
\lim_{n \to \infty} \| \pi - \pi_n \|_{ V_m } = 0,
\end{align}
where the norm $\| \cdot \|_{V_m}$ is as defined in \eqref{stationarydistributionnbound}. This convergence along with \eqref{uniformgnolybound} and the dominated convergence theorem, implies that in order to prove \eqref{sensconvrelationmain} it suffices to show point-wise convergence of $g_n$ to $g$ as $n \to \infty$, i.e.\
\begin{align}
\label{pointwisegnconv}
\lim_{n \to \infty} g_n(x) = g(x).
\end{align}
We shall establish this convergence with the help of the $\beta$-resolvent matrix (see \cite{Hart}), defined as
\begin{align*}
R_\beta = \beta ( \beta {\bf I} - Q )^{-1}
\end{align*}
where $\beta > 0$ and ${\bf I}$ is the identity matrix. It can be seen that $R_\beta$ is a positive matrix satisfying $R_\beta {\bf 1} = {\bf 1}$, $\pi^T R_\beta = \pi^T$ and
\begin{align}
\label{resolventidentifties}
R_\beta = {\bf I} + \beta^{-1} Q R_\beta = {\bf I} + \beta^{-1} R_\beta Q.
\end{align}
We can view $R_\beta$ as the probability transition matrix of a discrete time Markov chain over $\mathcal{E}$. The resolvent $\bar{R}_{\beta,n}$ corresponding to the projected CTMC with transition rate matrix $\bar{Q}_n$ and state-space $\mathcal{E}_n$, can be defined in a similar way.

Now select a state $x = x_i \in \mathcal{E}$ and let $\mu_x$ be the vector whose $i$-th coordinate is $1$ and the rest are all zeros. From the analysis in \cite{sFSP} we can conclude that for any positive integer $m$, there exist constants $C > 0$ and $\rho \in (0,1)$ such that 
\begin{align}
\label{mainresolventidentity}
\| \mu^T_x  R^{j}_\beta - \pi^T  \|_{V_m} \leq C V_m(x) \rho^j \qquad \textnormal{and} \qquad \| \mu^T_x  \bar{R}^{j}_{\beta,n} - \pi_n^T  \|_{V_m} \leq C V_m(x) \rho^j,
\end{align}
where $A^j$ denotes the $j$-th power of matrix $A$. Viewing functions $f$ and $g$ as vectors, we can rewrite Poisson equation \eqref{defn_poiss_eqn1} as
\begin{align*}
-Q g = f - \langle f, \pi \rangle {\bf 1}.
\end{align*}
Multiplying this equation by $R_\beta$ on the left and using \eqref{resolventidentifties} we obtain
\begin{align*}
({\bf I} - R_\beta) g = \beta^{-1} (f - \langle f, \pi \rangle {\bf 1}),
\end{align*}
which allows us to express $g(x)$ as
\begin{align}
\label{solnpoisseqn1}
g(x) = \mu^T_x g = \beta^{-1} \sum_{j=0}^\infty \mu^T_x R^j_\beta (f - \langle f, \pi \rangle {\bf 1}) =  \beta^{-1} \sum_{j=0}^\infty (\mu^T_x R^j_\beta  - \pi^T )f. 
\end{align}
This expression is well-defined due to \eqref{mainresolventidentity} and the fact that $f$ is polynomially growing. Moreover one can verify that $\langle g,\pi \rangle =0$. Similar to \eqref{solnpoisseqn1} we can represent the solution $g_n$ of the Poisson equation \eqref{defn_poiss_eqn2} as
\begin{align}
\label{solnpoisseqn2}
g_n(x) =   \beta^{-1} \sum_{j=0}^\infty (\mu^T_x \bar{R}^j_{\beta,n}  - \pi^T_n )f_n. 
\end{align}
Since \eqref{stationarydistributionconvmnorm} holds, the point-wise convergence of $f_n$ to $f$ as $n \to \infty$ is immediate and the convergence of the resolvent operator $\bar{R}_{\beta,n}$ to $R_{\beta}$ is proved in \cite{Hart}. Also notice that for some positive integer $m$, we have $\| f \|_{V_m} <\infty$ and $\sup_{n} \| f_n\|_{V_m} < \infty$, where the norm $\| \cdot \|_{V_m}$ is defined as in \eqref{defn_v_norm}. Hence using \eqref{mainresolventidentity} and the dominated convergence theorem we can conclude that \eqref{pointwisegnconv} holds for each state $x \in \mathcal{E}$ and this concludes the proof of this proposition. 
\end{proof}

\subsection{Solving the Poisson Equation \eqref{defn_poiss_eqn2}} \label{subsec:Poisson}

We now discuss how one can efficiently compute the steady-state sensitivity estimate $S_{n,\theta}(f) $ \eqref{approx_sens_formula}. We can obtain the approximation $\pi_{n,\theta}$ of the stationary distribution via sFSP \cite{sFSP}. Thereafter to compute $S_{n,\theta}(f) $ we need to find a solution to the finite-dimensional Poisson equation \eqref{defn_poiss_eqn2}. Note that this equation has a unique solution $g_n$, if we impose that $\langle g_n, {\bf 1} \rangle = 0$. We now focus on how this unique solution can be computed. For convenience, we drop the subscript $\theta$ from $\bar{Q}_{n,\theta}$ and denote this matrix as $\bar{Q}_n$ in this section. Observe that matrix $\bar{Q}_n$ is singular (as $\textnormal{Rank}(\bar{Q}_n) = n-1$) and usually very large in size. This poses difficulties in solving \eqref{defn_poiss_eqn2} with standard large-scale linear solvers that are currently in use (see Section \ref{subsec:code}). To overcome these issues, we solve the Poisson equation \eqref{defn_poiss_eqn2} using the \emph{Basis Function Method} (BFM) which we describe next. This method is approximate in nature, but it is computationally very feasible and it often yields a very accurate estimate of the solution of the Poisson equation. 

The idea of the BFM is very simple. We start by choosing a \emph{basis} of $m$ vectors (with $m \ll n$) $\mathcal{B} = \{g_1,\dots, g_m\}$ in $\R^n$. Then we construct the image of this basis under matrix $\bar{Q}_n$, by computing
\begin{align*}
f_i = - \bar{Q}_n g_i \quad \textnormal{for each} \quad i=1,\dots,m.
\end{align*}
Letting $ \bar{Q}_n (\mathcal{B}):= \{f_1,\dots, f_m\}$, our goal is to \emph{project} the r.h.s. of \eqref{defn_poiss_eqn2}, given by
\begin{align}
\label{rhsofpe}
\bar{f}_n = f_n - \langle f_n, \pi_{n,\theta} \rangle {\bf 1} 
\end{align}
on the linear space spanned by the vectors in $\bar{Q}_n( \mathcal{B})$. If this projection is given by
\begin{align}
\label{defn_hatf}
\hat{f}(c) = \sum_{i=1}^m c_i f_i
\end{align}
for some coefficient vector $c = (c_1,\dots,c_m) \in \R^m$, then our solution approximation is given by
\begin{align}
\label{desiredsolnapproxn}
\hat{g}(c) = \sum_{i=1}^m c_i g_i  -  \left( \sum_{i=1}^m c_i  \langle g_i, {\bf 1} \rangle \right) {\bf 1}.
\end{align}

To obtain the \emph{best} solution approximation, we construct the optimal projection $\hat{f}(c)$ by choosing the coefficient vector $c$ that minimizes the following \emph{error} function
\begin{align}
\label{defn_error_func}
\epsilon(c) = \left\| \pi_{n,\theta} \odot \left( \bar{f}_n - \hat{f}(c) \right) \right\|_{\ell_2},
\end{align}
where $\|\cdot\|_{\ell_2}$ denotes the standard $\ell_2$-norm and $\odot$ is the Hadamard (element-wise) vector product. In this error function, the discrepancy between $\bar{f}_n $ and its projection $\hat{f}(c)$ is weighted by the approximate stationary distribution $\pi_{n,\theta}$. This ensures that our projection is biased towards minimizing discrepancies in states with high stationary probabilities, which is important because contribution from these states is more dominant in the expression \eqref{approx_sens_formula} for the steady-state sensitivity $S_{n,\theta}(f)$.

Let ${\bf F} = \textnormal{Col}(f_1,\dots, f_m)$ be the $n \times m$ matrix whose columns are vectors in $\bar{Q}_n (\mathcal{B})$. Also let ${\bf \Pi}$ be the $n \times n$ diagonal matrix whose entries are given by the approximate stationary distribution $\pi_{n,\theta}$. The minimization of the error function $\epsilon(c)$ is essentially a weighted least-square problem whose optimal solution $c^*$ satisfies
\begin{equation}
\label{weightedLS}
{\bf A} c^*=  {\bf b}
\end{equation}
where ${\bf A }$ is the $m \times m$ matrix and ${\bf b}$ is the $m \times 1$ vector given by
\begin{align*}
{\bf A }= {\bf F^T  \Pi  F} \qquad \textnormal{and} \qquad {\bf b } = {\bf F^T \Pi } \ \bar{f}_n
\end{align*}
respectively. Since $m \ll n$, the $m$-dimensional linear-system \eqref{weightedLS} is much easier to solve than the Poisson equation \eqref{defn_poiss_eqn2}. If the optimal error $\epsilon( c^*)$ is below a small tolerance level, we can expect $\hat{g}(c^*)$ (see \eqref{desiredsolnapproxn}) to be an accurate solution of the Poisson equation \eqref{defn_poiss_eqn2}. On the other hand, if the error $\epsilon( c^*)$ is above the tolerance level, it can be reduced by expanding the set of basis functions $\mathcal{B}$ by including more vectors.  

Observe that the $n \times n$ matrix $\bar{Q}_n$ is \emph{extremely} sparse, as each row can have up to $(K+1)$ non-zero entries, where $K$ is the total number of reactions. The BFM essentially substitutes the problem of solving the Poisson equation \eqref{defn_poiss_eqn2} for a large sparse matrix, with the problem of solving a much smaller (but dense) linear system \eqref{weightedLS}. For this approach to work well, it is important that the set of basis vectors $\mathcal{B} = \{g_1,\dots, g_m\}$ is chosen in such a way that the linear span of the image set $ \bar{Q}_n ( \mathcal{B} )= \{f_1,\dots, f_m\}$, contains the vector $\bar{f}_n$ \eqref{rhsofpe}. While it is difficult to determine the basis vectors that meet this criterion, we find that for most biological examples of interest, this containment condition is fulfilled with a very small approximation error when we choose these basis vectors as representatives of \emph{monomials} of the form
\begin{align}
\label{defn_mon}
G(x_1,\dots,x_M) = \prod_{i=1}^M x^{e_i}_i 
\end{align}
where exponents $e_1,\dots,e_m$ are non-negative integers. Given such a monomial, its representative on the state-space $\mathcal{E}_n = \{x_{j_1},\dots, x_{j_n}\}$ is the $n$-dimensional vector $g= (G(x_{j_1}),\dots,G(x_{j_n})  )$. Hence for any positive integer $d_{ \textnormal{max}}$ specifying the maximal monomial degree, let $\mathcal{B}(d_{ \textnormal{max}})$ be the set of all basis vectors representing monomials with exponents $e_1,\dots,e_M$ satisfying $\sum_{i=1}^M e_i \leq d_{ \textnormal{max} }$. Letting $\binom{a}{b}:=\frac{a!}{b! (a-b)!}$ denote the binomial coefficient, the number of basis vectors in $\mathcal{B}(d_{ \textnormal{max}})$ is equal to
\begin{align*}
|\mathcal{B}(d_{ \textnormal{max}})| = \sum_{d=1}^{d_{ \textnormal{max}}} \binom{M-1 +d}{d},
\end{align*} 
which grows rapidly as the maximal degree $d_{ \textnormal{max} }$ and the number of species $M$ increases.

We refer to our steady-state sensitivity estimation method as the \emph{Poisson Estimator} (PE) as it is based on solving the Poisson Equation \eqref{defn_poiss_eqn2} associated with the stochastic reaction network. The main steps required by our method are summarized in Algorithm \ref{alg:pe}.

\begin{algorithm}[H]
\caption{{\bf Poisson Estimator (PE)}: Provides an estimate of the steady-state sensitivity \eqref{defn_ssparamsens}.}
\label{alg:pe}
\begin{algorithmic}[1]
\Require Reaction stoichiometries, propensities, parameters, and objective function $f$.
\Ensure Sensitivities of $f$ w.r.t. every parameter.
\Procedure{PE}{}
\State Set $n \leftarrow 0$ and $\mathcal{E}_0 := \emptyset$
\State Set $n \leftarrow n+1$ and choose a truncated state space $\mathcal{E}_n$ with $|\mathcal{E}_n| > |\mathcal{E}_{n-1}|$
\State Run sFSP to estimate stationary distribution $\pi_{\theta}$. Jump to step 3 and repeat with a larger truncated state-space if necessary.
\State Select $d_{\max}$ and construct the basis vectors
\State Solve the Poisson equation \eqref{defn_poiss_eqn2} by minimizing \eqref{defn_error_func}.
\If {the BFM error is too large}
\State Increase $d_{\max}$ and jump to step 5.
\Else \State Compute sensitivities of $f$ w.r.t. every parameter using the formula \eqref{approx_sens_formula}.
\State \Return
\EndIf
\EndProcedure
\end{algorithmic}
\end{algorithm}

\subsection{Software Implementation}
\label{subsec:code}

We have computationally implemented our method (PE) as a C++ library, called \emph{cossmosLib}, to which one can specify an arbitrary reaction network (as a script in \emph{Systems Biology Markup Language} (SMBL) \cite{hucka2003systems}) and output estimates of the desired steady-state parametric sensitivities. In our software implementation, the most computationally intensive tasks are \emph{parallelized} via support of the Message Passing Interface (MPI). The \emph{cossmosLib} can be downloaded from our \href{https://git.bsse.ethz.ch/patrikd/cossmosLib}{GitLab} repository \cite{cossmosLib}. This page contains full documentation of the library along with detailed installation instructions.

Our library has several dependencies on (mostly) open source software. The only non-open source software we employ, are the Intel MPI and Intel MKL (Math Kernel Library) libraries. Note that these libraries could in theory be interchanged with corresponding open source libraries (some implementations of the BLAS and LAPACK libraries as well as some implementation, like Open MPI, of the MPI library), but in our experience the Intel libraries performed the best. Furthermore we employ the ParMETIS library \cite{Karypis98} from the Karypis Lab, the SuperLU\_dist library \cite{Li03} from the Computational Research Division of the Berkeley Lab Computing Sciences, the libSBML library \cite{Bornstein08} from the SBML community and the Trilinos library \cite{Baker2009,Heroux2003,Sala07} from the Sandia National Laboratories.

\section{Examples} \label{sec:examples}

In this section we test our method for steady-state sensitivity estimation on several examples, and compare it with the simulation method IntCLR \cite{wang2018steady} for estimating the steady-state sensitivity. In the preceding sections we discussed sensitivity computation w.r.t.\ a scalar parameter $\theta$, but the extension to a vector-valued parameter is straightforward. For vector-valued $\theta = (\theta_1,\theta_2,\dots)$, we view the steady-state sensitivity value $S_\theta(f)$ as the vector
\begin{align*}
S_\theta(f)  =( S_{\theta_1}(f), S_{\theta_2}(f),\dots  ),
\end{align*}
where each $S_{\theta_i}(f)$ is defined as \eqref{defn_ssparamsens}. Suppose we obtain an estimate $\hat{S}_{\theta_i}(f)$ of the true sensitivity value $S_{\theta_i}(f)$. Then the percentage \emph{Relative Error} ({\bf RE}) is defined as
\begin{align}
\label{defn_RE_formula}
\textnormal{RE}(\theta_i) = \left\{
\begin{array}{cc}
\left| \frac{ \hat{S}_{\theta_i}(f)-  S_{\theta_i}(f)}{ S_{\theta_i}(f)}  \right| \times 100  & \quad \textnormal{if} \quad S_{\theta_i}(f) \neq 0  \\
|  \hat{S}_{\theta_i}(f) | & \quad \textnormal{otherwise}.
\end{array}\right.
\end{align}
Moreover we define the percentage \emph{Total Relative Error} ({\bf TRE}) as
\begin{align*}
\textnormal{TRE}(\theta) = \sum_{i} \textnormal{RE}(\theta_i).
\end{align*}
From now on we denote the steady-state expectation of any output function $f$ of our $\theta$-dependent stochastic reaction network as $\E_\theta(f)$. In all our examples $f$ takes the form $f(x) =x_i$, i.e. the output is simply the copy-number of species $\mathcal{S}_i$. Therefore we shall express the steady-state expectation as $\E_\theta(x_i)$ and the corresponding $\theta$-sensitivity vector as $S_\theta(x_i)  =( S_{\theta_1}(x_i), S_{\theta_2}(x_i),\dots  )$.

In our first three examples, the irreducible state-space is $\N^d_0$, where $d$ is the number of species. For our computations, we shall employ trapezoidal truncations of the state-space (see \cite{sFSP}) which are defined using two cut-off values $C_l$ and $C_r$ (with $C_l \leq C_r$) as  
\begin{align}
\label{defn_trap_trunc}
\mathcal{T}(C_l,C_r) = \{ x \in \mathbb{N}^d_0:\phi_d(\boldsymbol{0},C_l) \leq \phi_d(x) \leq \phi_d(C_r,\boldsymbol{0})\}.
\end{align}
Here we follow the convention that $ \binom{a}{b} = 0$ if $a < b$, and the function $\phi_d$ is given by  
\begin{align*}
\phi_d(x_1,\dots, x_d) = \binom{x_1+\cdots+x_d+d-1}{d} + \binom{x_1+\cdots+x_{d-1}+d-2}{1} + \cdots + \binom{x_1}{1}.
\end{align*}
In the last example we consider, the irreducible state-space is a finite set and hence no truncations are necessary for the computations.

Finally we would like to point out that all the computations were performed on the \emph{Euler} cluster of ETH Zurich\cite{Euler}. This allowed us to test the parallel implementation of our library \emph{cossmosLib}. 
 
\subsection{Gene Expression Network}
\label{subsec:genex}

Our first example is the standard gene-expression network introduced in \cite{MO} (see Figure \ref{fig::genex}(A)). It consists of two species $\mathcal{S}_1$ (mRNA) and $\mathcal{S}_2$ (Protein) which participate in four reactions whose propensities are parametrized by $\theta = (\theta_1,\theta_2,\theta_3,\theta_4)$. These reactions are described in Table \ref{genex} and we set the parameters as $\theta_1 = 90$, $\theta_2 = 4$, $\theta_3 = 0.5$ and $\theta_4 = 0.2$. 

We are interested in estimating the sensitivity-vectors $S_\theta(x_1)$ and $S_\theta(x_2)$, which correspond to parametric sensitivities of the steady-state expectation of the mRNA copy-numbers ($\E_\theta (x_1)$) and the protein copy-numbers ($\E_\theta (x_2)$) respectively. In this example all propensity functions are linear, and hence the steady-state expectations $\E_\theta (x_1)$ and $\E_\theta (x_2)$ can be explicitly computed as functions of $\theta$. Consequently the sensitivity vectors can also be exactly evaluated and used for estimating the RE.

We next estimate the sensitivity-vectors with our PE method given by Algorithm \ref{alg:pe}. We employ our method with five trapezoidal state-space truncations mentioned in Table \ref{table_ge_state_spaces}. This table also mentions the designated state chosen for sFSP. For estimating the solution of the finite-dimensional Poisson equation \eqref{defn_poiss_eqn2} via BFM we use $65$ basis vectors, which corresponds to $d_{ \textnormal{max} } =10$ (see Section \ref{subsec:Poisson}). The sensitivity-vectors estimated by PE, with the largest truncated  state-space, are reported in Table \ref{tbl:genexSens} and depicted as a bar chart in Figure \ref{fig::genex}(B) along with their percentage \emph{Relative Error} (RE). Note that the RE$\%$ is always below $3 \cdot 10^{-7}$ which demonstrates the accuracy of PE method. In Figure \ref{fig::genex}(C) we illustrate how the \emph{Total Relative Error} (TRE) decreases, while the memory requirements get larger as the truncated state-space expands in size. Even though the memory requirements increase linearly with the state-space size, the decrease in the TRE$\%$ is almost exponential, which is consistent with the results reported in \cite{sFSP}. In Figure \ref{fig::genex}(D) we highlight potential drawbacks of the simulation-based approach IntCLR in estimating steady-state sensitivity $S_{\theta_1}(x_2)$. Here the total central processing time is fixed to be $120$ hours, and the final time $T$ denotes the time at which steady-state is assumed to be reached and at which the sensitivity is computed. One can see that as $T$ increases, the estimation accuracy improves (i.e. RE decreases) but the number of simulation samples that are generated within 120 hours gets lower and hence the standard deviation of the estimator increases. Such a trade-off is natural to expect for any simulation-based method for estimating steady-state sensitivity. Note that even though our PE method has high memory requirements, even with the largest truncated state-space we consider, PE yields all the estimates within just $19$ minutes (with four processors on ETH Zurich's Euler cluster) which is much smaller than the processing time for IntCLR (120 hours). Moreover even the most accurate estimate obtained by IntCLR has a RE$\%$ 10,000 times higher than the RE$\%$ for the corresponding PE-estimated value.

\begin{figure}[h]
\includegraphics[width=\textwidth]{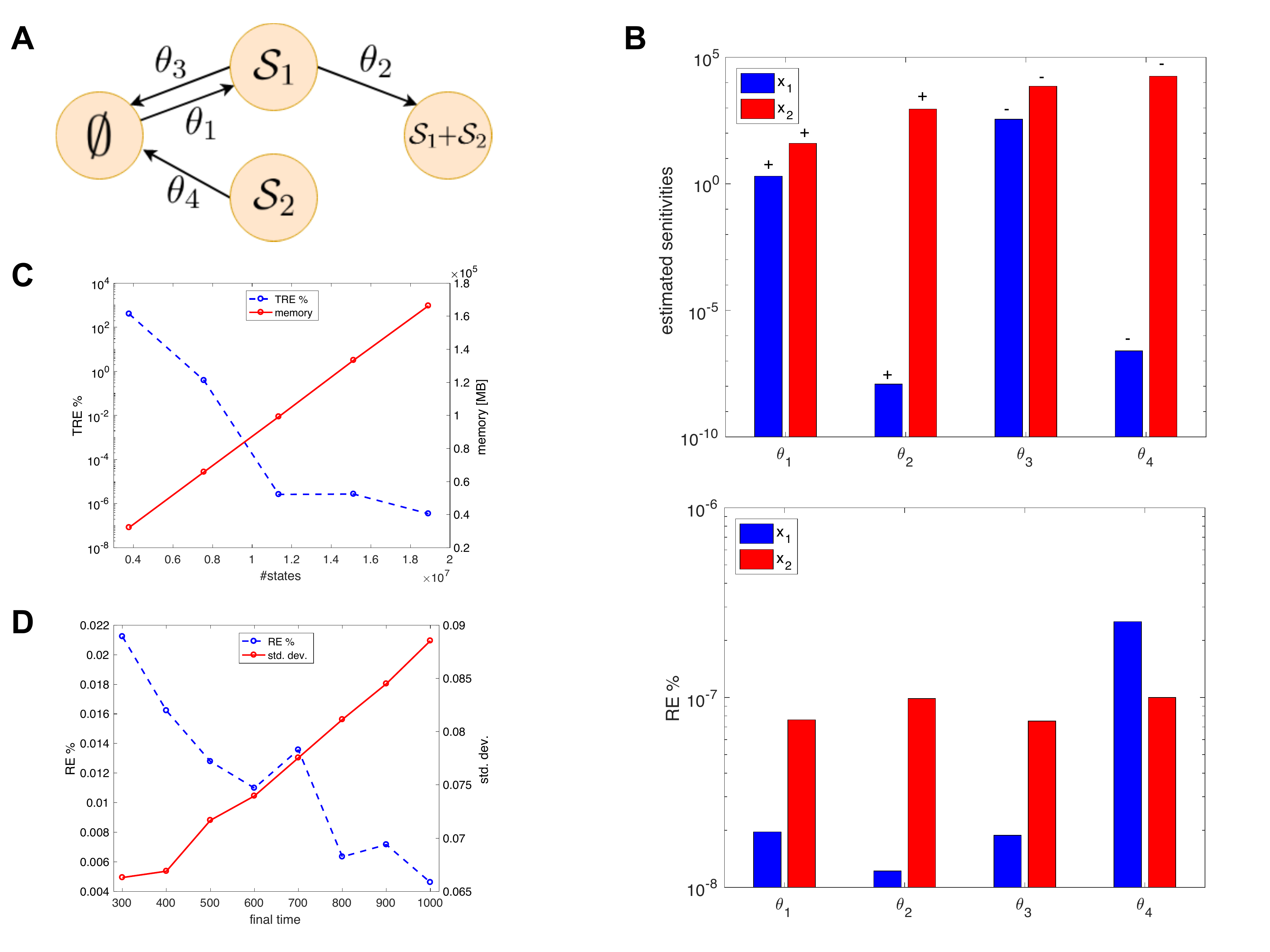}
\caption{ \scriptsize 
Sensitivity analysis of the gene-expression network depicted in (A). Species $\mathcal{S}_1$ and $\mathcal{S}_2$ represent mRNA and Protein respectively. All reactions have mass-action propensities and the vector of rate constants is $\theta = (\theta_1,\theta_2,\theta_3,\theta_4) = (90,4,0.5,0.2)$. (B) Bar graphs showing the PE-estimated sensitivity values (top) along with the corresponding RE$\%$ (bottom). Sensitivities are computed w.r.t. all parameters for two objective functions ($x_1$ and $x_2$). Note that in the bar graph only the absolute value of the sensitivity value is plotted while its sign is displayed on the top of the bar. (C) Plots the TRE$\%$ for the sensitivity vector $S_\theta(x_2) = ( S_{\theta_1}(x_2), S_{\theta_2}(x_2), S_{\theta_3}(x_2),S_{\theta_4}(x_2))$ for different state-space truncations listed in Table \ref{table_ge_state_spaces}. The total memory used (by all the processors) is also plotted. (D) Plots the percentage relative error (RE$\%$) for the sensitivity $\mathcal{S}_{\theta_1}(x_2)$ estimated with the simulation-based estimation method IntCLR \cite{wang2018steady} implemented on a single processor with total processing time of $120$ hours. Each estimate was produced by simulating the stochastic trajectories in the time-period $[0,T]$, where $T$ is the final time at which the steady-state is assumed to be reached. The statistical accuracy of the estimate (measured as standard deviation of the estimator) is also shown. } 
\label{fig::genex}
\end{figure}

\subsection{Toggle-Switch Network}
\label{subsec:toggle}

Our next example is the synthetic toggle-switch proposed by Gardner et al.\cite{Gardner}. Here two species $\mathcal{S}_1$ and $\mathcal{S}_2$ are simply \emph{repressing} the production of each other. This network consists of four reactions whose propensities are parametrized by $\theta = (\theta_1,\dots,\theta_6)$ (see Table \ref{ToggleSwitch}). We set the parameters as $\theta_1 = 500$, $\theta_2 = 3$, $\theta_3 = 200$, $\theta_4 = 0.4$, $\theta_5 = 1.5$ and $\theta_6 = 1$. As in the previous example, we are interested in estimating the sensitivity-vectors $S_\theta(x_1)$ and $S_\theta(x_2)$, which correspond to parametric sensitivities of the steady-state expectation of the copy-numbers of species $\mathcal{S}_1$ and $\mathcal{S}_2$ respectively. It is well-known that the stationary distribution for the CTMC model of the toggle-switch is \emph{bimodal}, with each mode corresponding to one of the species being dominant \cite{sFSP}. This bimodality causes problems in estimating the stationary distribution and it arises due to nonlinearities in the propensity functions. Moreover these nonlinearities prevent us from determining the steady-state sensitivities exactly. Hence we cannot compute the RE$\%$ for this example.

We now estimate the sensitivity-vectors with our PE method (Algorithm \ref{alg:pe}). We employ six trapezoidal state-space truncations (see Table \ref{table_tgs_state_spaces}), with the designated state set as $x_\ell = (235,115)$. We choose $d_{ \textnormal{max} } =10$ to obtain $65$ basis vectors for the BFM. The sensitivity-vectors estimated by PE, with the largest truncated  state-space, are reported in Table \ref{tbl:TSSens} and depicted as a bar chart in Figure \ref{fig::tgs}(B). To obtain these PE-estimated sensitivity values, overall 13 minutes were required with a single processor on ETH Zurich's Euler cluster. Figure \ref{fig::tgs}(C) shows that as expected, the increase in memory requirements is linear in the truncated state-space size. Next we estimated a single steady-state sensitivity $\mathcal{S}_{\theta_5}(x_2)$ using IntCLR with total processing time of $120$ hours. As the true sensitivity value is unknown, we estimate the accuracy of the simulation-based estimate by computing eRE$\%$ obtained by substituting the true sensitivity value with the PE-estimated value in \eqref{defn_RE_formula}. Like in the previous example, as the simulation time-period increases, the estimate becomes more accurate (i.e.\ eRE$\%$ decreases) but the statistical accuracy of the estimate (measured as standard deviation of the estimator) deteriorates as the number of simulation samples obtained in the limited processing time gets smaller (see Figure \ref{fig::tgs}(D)).

\begin{figure}[h]
	\includegraphics[width=0.8\textwidth]{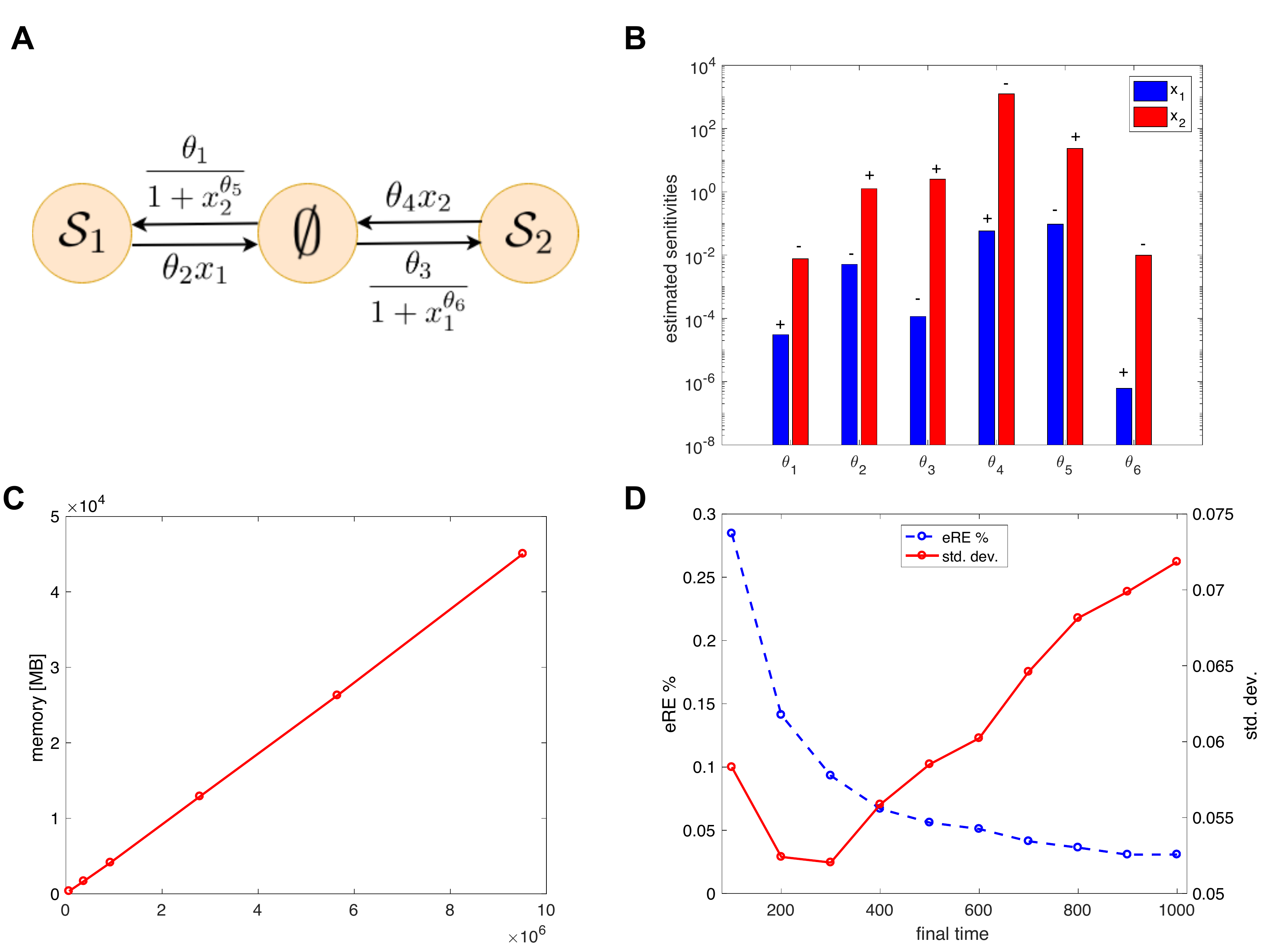}
	\caption{
	\scriptsize 
Sensitivity analysis of the toggle-switch network depicted in (A). Species $\mathcal{S}_1$ and $\mathcal{S}_2$ mutually repress each others production which is captured by a Hill function. The propensity functions are parametrized by 
$\theta = (\theta_1,\dots,\theta_6) = (500,3,200,0.4,1.5,1)$. (B) Bar graphs showing the PE-estimated sensitivity values computed w.r.t. all parameters for two objective functions ($x_1$ and $x_2$). The sign above the bars indicates whether the respective estimate is positive or negative. (C) Plots the memory requirements for different state-space truncations listed in Table \ref{table_tgs_state_spaces}. D) Plots the estimated percentage relative error (eRE$\%$) for the sensitivity $\mathcal{S}_{\theta_5}(x_2)$ estimated with the simulation-based estimation method IntCLR \cite{wang2018steady} implemented on a single processor with total processing time of $120$ hours. Each estimate was produced by simulating the stochastic trajectories in the time-period $[0,T]$, where $T$ is the final time at which the steady-state is assumed to be reached. The statistical accuracy of the estimate (measured as standard deviation of the estimator) is also shown.}
	\label{fig::tgs}
\end{figure}

\subsection{A Deficiency Zero Network} 
\label{subsec:DefZero}

We now consider a larger network where the stationary distribution is explicitly known and hence the steady-state sensitivities can be computed explicitly. Consider the four-species network depicted in Figure \ref{fig::dz}(A) whose reactions are described in Table \ref{DefZero}. The parameter vector for this network is $\theta = (\theta_1,\dots,\theta_6) = (4.5, 0.8,5,1,0.6,11,3,80)$. In this example we are interested in estimating the sensitivity vectors for all the species (i.e.\ $S_{\theta}(x_i)$ for $i=1,2,3,4$). Note that each reaction $k$ of the form \eqref{reactionform}, can be represented as $\nu_k \longrightarrow \nu'_k$, where $\nu_k = (\nu_{1k},\nu_{2k},\dots)$ and $\nu'_k = (\nu'_{1k},\nu'_{2k},\dots)$ are non-negative integer vectors representing the \emph{source} complex and the \emph{product} complex, respectively. The network under consideration is \emph{weakly-reversible} (i.e. for each reaction $\nu_k \longrightarrow \nu'_k$ there is a sequence of reactions realizing $\nu'_k \longrightarrow \nu_k$) and it has \emph{deficiency zero} \footnote{The deficiency of a reaction network is defined as $\delta = C - l -s$, where $C$ is the number of complexes, $l$ is the number of connected components in the complex reaction graph (see Figure \ref{fig::dz}(A) for example) and $s$ is the dimension of the subspace spanned by the stoichiometry vectors $\zeta_k$-s.}. Moreover all the propensities have mass-action form \eqref{massactionkin_defn}. From Theorem 4.2 in \cite{anderson2010product} it can be seen that the stationary distribution for this network has the product-form Poisson distribution
\begin{align*}
\pi_\theta ( x_1,x_2,x_3,x_4 ) = \exp\left( -\sum_{i=1}^4 c_i(\theta) \right) \prod_{i=1}^4 \frac{ (c_i(\theta) )^{x_i}}{x_i!},
\end{align*}
where $c(\theta) = (c_1(\theta), c_2(\theta) , c_3(\theta) , c_4(\theta)  )$ is the mean vector given by
\begin{align*}
c_1(\theta) = \frac{\theta_1}{ \theta_2 }, \quad c_2(\theta) = \frac{\theta_3}{ \theta_4}, \quad c_3(\theta) = \frac{\theta_1 \theta_3 \theta_7}{ \theta_2 \theta_4 \theta_8} \quad \textnormal{and} \quad c_4(\theta) = \frac{\theta^2_3 \theta_5}{ \theta^2_4 \theta_6}.
\end{align*}
Using this explicit form of the stationary distribution we can exactly compute the sensitivity vectors at our chosen value of parameter-vector $\theta$. These exact values allow us to compute RE$\%$ and assess the accuracy of our sensitivity estimates.

\begin{figure}[h]
	\includegraphics[width=0.8\textwidth]{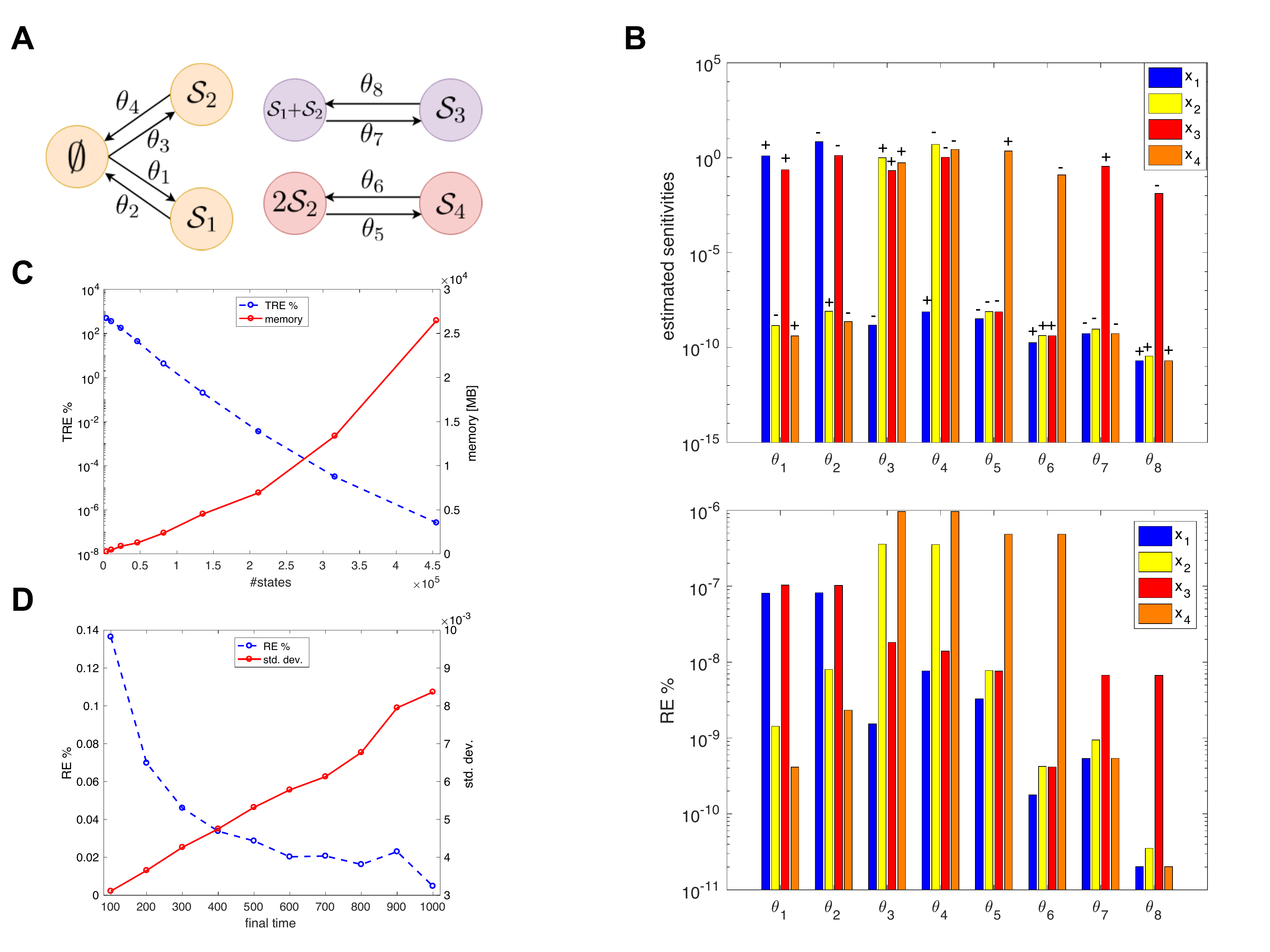}
	\caption{ \scriptsize 
Sensitivity analysis of the deficiency zero network depicted in (A). All reactions have mass-action propensities and the vector of rate constants is $\theta = (\theta_1,\cdots,\theta_8) = (4.5,0.8,5,1,0.6,11,3,80)$. (B) Bar graphs showing the PE-estimated sensitivity values (top) along with the corresponding RE$\%$ (bottom). Sensitivities are computed w.r.t. all parameters for two objective functions ($x_1$ and $x_2$). Note that in the bar graph only the absolute value of the sensitivity value is plotted while its sign is displayed on the top of the bar. (C) Plots the TRE$\%$ for the sensitivity vector $S_\theta(x_3) = ( S_{\theta_1}(x_2),
\dots,S_{\theta_8}(x_2))$ for different state-space truncations listed in Table \ref{table_dz_state_spaces}. The total memory used is also plotted. (D) Plots the percentage relative error (RE$\%$) for the sensitivity $\mathcal{S}_{\theta_5}(x_4)$ estimated with the simulation-based estimation method IntCLR \cite{wang2018steady} implemented on a single processor with total processing time of $120$ hours. Each estimate was produced by simulating the stochastic trajectories in the time-period $[0,T]$, where $T$ is the final time at which the steady-state is assumed to be reached. The statistical accuracy of the estimate (measured as standard deviation of the estimator) is also shown. } 
	\label{fig::dz}
\end{figure}

We now employ our PE method (Algorithm \ref{alg:pe}) to estimate the sensitivity-vectors. We use nine trapezoidal state-space truncations (see Table \ref{table_dz_state_spaces}) and set the designated state as $x_\ell = (10,0,0,0)$. We choose $d_{ \textnormal{max} } =10$ to obtain $1000$ basis vectors for the BFM. The sensitivity-vectors estimated by PE, with the largest truncated state-space, are reported in Table \ref{tbl:defZeroSens} and depicted as a bar chart in Figure \ref{fig::dz}(B). Note that for all these sensitivity estimates the RE$\%$ is less than $10^{-6}$ certifying the accuracy of our PE method. All these estimations took 73 minutes with a single processor on ETH Zurich's Euler cluster. Like in the first example, as the size of the truncated state-spaces increases, the overall error (TRE$\%$) decreases exponentially but the memory requirements increase linearly (see Figure \ref{fig::dz}(C)). Next a single steady-state sensitivity $S_{\theta_5}(x_4)$ was estimated using IntCLR with total processing time of $120$ hours. As before, increasing the simulation time-period reduces the error but worsens the statistical accuracy of the estimator (see Figure \ref{fig::dz}(D)). Remarkably, even the most accurate estimate has a RE$\%$ which is 1'000 times higher than the RE$\%$ for the corresponding PE-estimated value.

\subsection{A Simple Linear Network} 
\label{subsec:paper}

As our last example, we consider the simple linear network given in \cite{wang2018steady} involving three species and four conversion reactions (see Figure \ref{fig::sln}(A) and Table \ref{PaperSystem}). The vector of reaction-rate constants is $\theta = (\theta_1,\theta_2,\theta_3,\theta_4)$ and we set the parameters as $\theta_1 = 10,\theta_2 = 20, \theta_3 = 0.03$ and $\theta_4 = 0.02$. We are interested in estimating all sensitivities of the form $S_{\theta_i}(x_j)$ for $i=1,\dots,4$ and $j=1,\dots,3$. Exploiting the linearity of the propensity function, we can analytically compute the steady-state expectations and the exact sensitivity values. These values are used for estimating the RE for our sensitivity estimates.

Notice that in this reaction network there is no inflow or outflow of mass and the total number of molecules of all the species is conserved, i.e. the CTMC dynamics for this reaction network evolves on the finite state-space
\begin{align*}
\mathcal{E} = \{ x = (x_1,x_2,x_3) \in \N^3_0 :  x_1 + x_2 + x_3 = c\}
\end{align*}
where the constant $c$ is determined by the initial condition. We choose $c = 10$ and for this value the state-space $\mathcal{E}$ is reasonably small and sFSP can yield the exact stationary distribution $\pi_\theta$ without requiring any state-space truncation. Hence when we estimate the sensitivity values using our PE method (Algorithm \ref{alg:pe}) the main source of error is the error incurred by BFM in solving the Poisson equation. It can be seen from Figure \ref{fig::sln}(B) and Table \ref{tbl:PaperSystemSens}, that this error is very small (RE$\%  \leq 2 \cdot 10^{-5}$ for all estimates) and the PE-based estimation is highly accurate. Here only two basis vectors were used for BFM which corresponds to $d_{ \textnormal{max} } = 1$. We then estimate one steady-state sensitivity value $\mathcal{S}_{\theta_3}(x_1)$ with the simulation-based approach IntCLR with the total central processing time fixed to $120$ hours. The results are shown in Figure \ref{fig::sln}(C) for different values of the final time $T$ where the steady-state is assumed to be reached. As in the previous examples, when $T$ increases the accuracy of the IntCLR estimate improves (i.e. RE$\%$ decreases), but the estimator standard deviation also increases due to lower number of simulation samples being generated in the limited processing time of 120 hours. However for this simple example, PE method yields all the sensitivity estimates within just 0.2 seconds (with a single processor on ETH Zurich's Euler cluster) and the accuracy of the PE estimate (measured by RE$\%$) was 900 times better than the best estimate provided by IntCLR.

\begin{figure}[h]
	\includegraphics[width=0.8\textwidth]{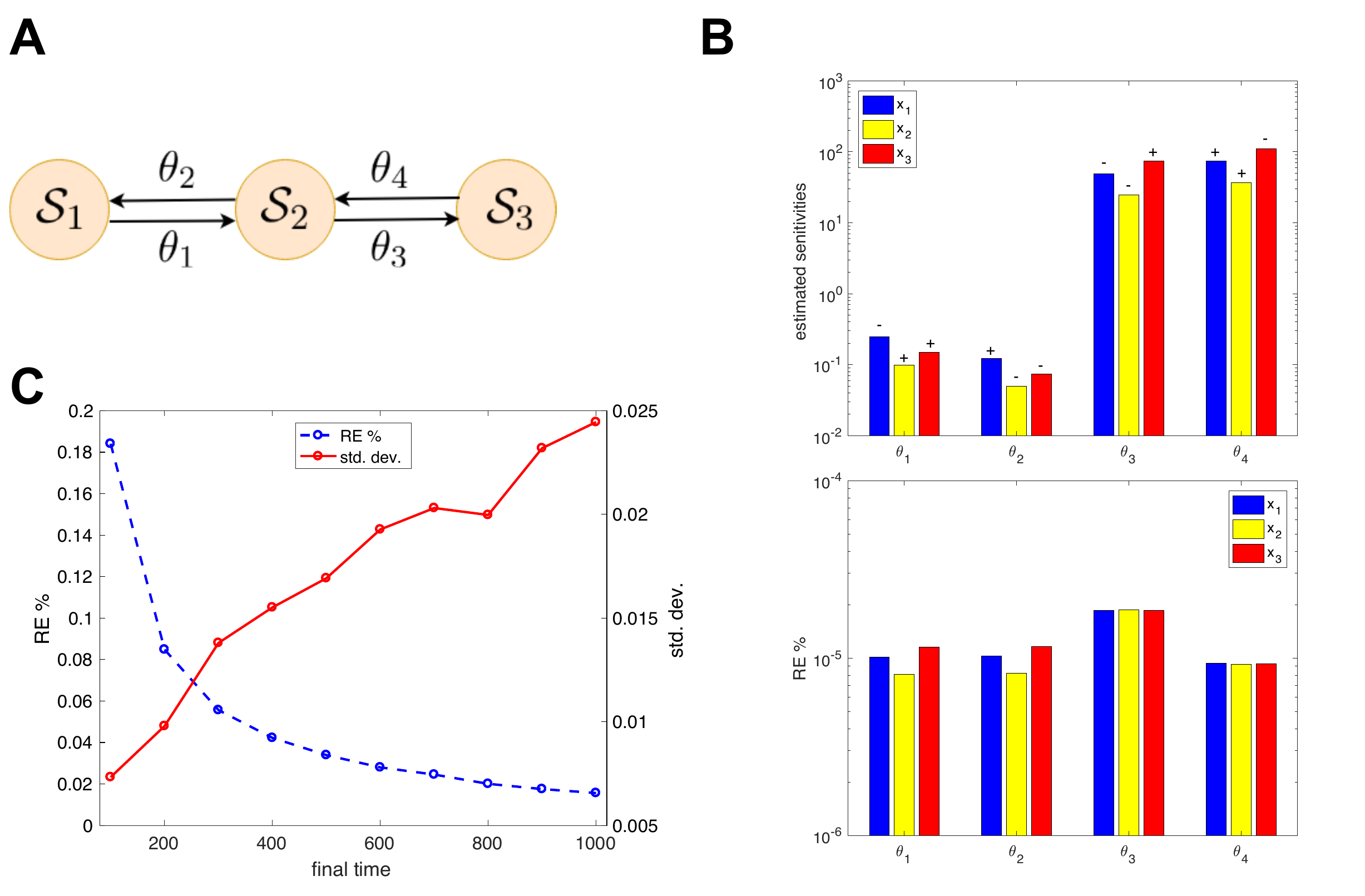}
\caption{ \scriptsize 
Sensitivity analysis of the simple linear network depicted in (A). All reactions have mass-action propensities and the vector of rate constants is $\theta = (\theta_1,\theta_2,\theta_3,\theta_4) = (10,20,0.03,0.02)$. (B) Bar graphs showing the PE-estimated sensitivity values (top) along with the corresponding RE$\%$ (bottom). Sensitivities are computed w.r.t. all parameters for three objective functions ($x_1$, $x_2$ and $x_3$). Note that in the bar graph only the absolute value of the sensitivity value is plotted while its sign is displayed on the top of the bar. (C) Plots the percentage relative error (RE$\%$) for the sensitivity $\mathcal{S}_{\theta_3}(x_1)$ estimated with the simulation-based estimation method IntCLR \cite{wang2018steady} implemented on a single processor with total processing time of $120$ hours. Each estimate was produced by simulating the stochastic trajectories in the time-period $[0,T]$, where $T$ is the final time at which the steady-state is assumed to be reached. The statistical accuracy of the estimate (measured as standard deviation of the estimator) is also shown.} 
\label{fig::sln}
\end{figure}

\section{Conclusion} \label{sec:conclusion}

The aim of this paper is to develop a numerical method for the estimation of the steady-state parameter sensitivity for stochastic reaction network models where the dynamics is described by a continuous-time Markov chain (CTMC) \cite{DASurvey}. Unlike other sensitivity estimation approaches \cite{IRN,Gir,KSR1,KSR2,DA,Our,Gupta2,wang2018steady}, our method does not require time-consuming simulations of the stochastic trajectories, but instead it relies on approximately solving a Poisson equation associated with the generator of the CTMC as well as on estimation of the stationary distribution obtained with the recently proposed \emph{stationary Finite State Projection} (sFSP) method \cite{sFSP}. We call our method the \emph{Poisson Estimator} (PE) and we mathematically prove the accuracy of this method under certain conditions. Moreover using many examples we demonstrate that it compares favorably to simulation-based methods, like IntCLR \cite{wang2018steady}, that exist for steady-state sensitivity estimation. In particular we found that estimation with PE requires far less computational time than with IntCLR and the estimates provided by PE are much more accurate than the corresponding IntCLR estimates. However PE has much higher memory requirements than IntCLR, and to circumvent this problem parallel computations may be necessary for even moderately sized networks.

In order to help the usability of our method we have developed a C++ library, called \emph{cossmosLib} \cite{cossmosLib}, that provides parallel implementation of the PE method for efficient estimation of both the stationary distribution (via sFSP) and the steady-state parameter sensitivities. Currently the parallelization in the \emph{cossmosLib} library is achieved by incorporating the \emph{Message Passing Interface} (MPI) and existing sparse matrix factorization methods that are part of the \emph{Trilinos} \cite{Baker2009,Heroux2003,Sala07} package. Further improvements in the performance of \emph{cossmosLib} can be obtained by customizing the parallel computations in such a way that the structure and the transition kernel of the CTMC are properly exploited in finding the stationary distribution and solving the Poisson equation.

%\appendix
\setcounter{equation}{0}

\renewcommand {\thetable}{A.\arabic{table}}
\section*{Appendix: Supplementary Tables}

This section contains the tables describing the reactions networks, sensitivity analysis and truncated state-spaces for all the examples.

      \begin{table}[h]
    \begin{center}
        \begin{tabular}{|c|c|c|c|}
            \hline
            No. & Reaction & Propensity \\
            \hline
            1 & $\emptyset \rightarrow \mathcal{S}_1$ & $\lambda_1(x,\theta) = \theta_1$ \\
            2 & $\mathcal{S}_1 \rightarrow \mathcal{S}_1+\mathcal{S}_2$ & $\lambda_2(x,\theta) = \theta_2 x_1$ \\
            3 & $\mathcal{S}_1 \rightarrow \emptyset$ & $\lambda_3(x,\theta) = \theta_3 x_1$ \\
            4 & $\mathcal{S}_2 \rightarrow \emptyset$ & $\lambda_4(x,\theta) = \theta_4 x_2$ \\
            \hline
        \end{tabular}
        \caption{Reactions in the gene-expression network. The propensities are parametrized by $\theta = (\theta_1,\theta_2,\theta_3,\theta_4)$.}
        \label{genex}
    \end{center}
\end{table}
\begin{table}[h]
\begin{center}
\begin{tabular}{|c|c|c|c|c|c|c|}
\hline
No. & \multicolumn{2}{|c|}{Cut-offs}  & State-space size & Designated state  \\
\hline
$i$ & $C_{l,i}$ & $C_{r,i}$ & $n_i$ & $x_\ell$  \\
\hline
\hline
$1$ & $3280$ & $4280$ & $3,784,784$ & $(10, 3270)$ \\
$2$ & $2780$ & $4780$ & $7,565,784$  & $(10, 2770)$ \\
$3$ & $2280$ & $5280$ & $11,346,784$  &  $(10, 2270)$ \\
$4$ & $1780$ & $5780$ & $15,127,784$ & $(10, 1770)$ \\
$5$ & $1280$ & $6280$ & $18,908,784$ & $(10, 1270)$ \\
\hline
\end{tabular}
\end{center}
\caption{Truncated state-spaces and designated states used in the gene-expression network. The cut-off values ($C_l$ and $C_r$) for the trapezoidal truncations \eqref{defn_trap_trunc} are provided.}
\label{table_ge_state_spaces}
\end{table} 
\begin{table}[h!]
    \begin{center} {\small
        \begin{tabular}{|c | c | c | c | c | c | c |}
            \hline
            Objective & Sensitive & Sensitivity & \multirow{2}{*}{RE\%} & Sensitive & Sensitivity & \multirow{2}{*}{RE\%} \\
            Function & Parameter & Values (PE) & & Parameter & Values (PE) & \\ \hline \hline
            \multirow{2}{*}{$x_1$} & $\theta_1$ & $2.0000$
            & $1.9609\cdot 10^{-8}$ & $\theta_3$ & $-360.0000$ & $1.8802\cdot 10^{-8}$\\
            & $\theta_2$ & $1.2186\cdot 10^{-8}$ & $1.2186\cdot 10^{-8}$ & $\theta_4$ & $-2.5096\cdot 10^{-7}$ & $2.5096\cdot 10^{-7}$\\ \hline \hline
            \multirow{2}{*}{$x_2$} & $\theta_1$ & $40.000$ & $7.6378\cdot 10^{-8}$ & $\theta_3$ & $-7200.0$ & $7.5121\cdot 10^{-8}$\\
            & $\theta_2$ & $900.00$ & $9.8708\cdot 10^{-8}$ & $\theta_4$ & $-18000$ & $9.9957\cdot 10^{-10}$\\ \hline \hline
        \end{tabular} }
    \end{center}
    \caption{Sensitivity analysis for the gene-expression network. This table lists the PE-estimated sensitivity values w.r.t. all four parameters $\theta_1,\dots,\theta_4$ for two objective functions $f(x) =x_1$ and $f(x) = x_2$ with expected steady-state values $\E_{\theta}(x_1) = 180$ and $\E_{\theta}(x_2)  = 3600$ respectively. The corresponding percentage relation error RE$\%$ \eqref{defn_RE_formula} for the sensitivity estimates is also provided. To perform all these computations the total processing time was $1142$ seconds on four processors and the total memory requirement was 166'278 MB.}
    \label{tbl:genexSens}
\end{table}

\begin{table}
    \begin{center}
        \begin{tabular}{|c|c|c|c|}
            \hline
            No. & Reaction & Propensity \\
            \hline
            1 & $\emptyset \rightarrow \mathcal{S}_1$ & $\lambda_1(x,\theta) = \frac{\theta_1}{1+x_2^{\theta_5}}$ \\
            2 & $\mathcal{S}_1 \rightarrow \emptyset$ & $\lambda_2(x,\theta) = \theta_2 x_1$ \\
            3 & $\emptyset \rightarrow \mathcal{S}_2$ & $\lambda_3(x,\theta) = \frac{\theta_3}{1+x_1^{\theta_6}}$ \\
            4 & $\mathcal{S}_2 \rightarrow \emptyset$ & $\lambda_4(x,\theta) = \theta_4 x_2$ \\
            \hline
        \end{tabular}
        \caption{Reactions in the Toggle-Switch network. The propensities are parametrized by $\theta = (\theta_1,\theta_2,\theta_3,\theta_4,\theta_5,\theta_6)$.}
        \label{ToggleSwitch}
    \end{center}
\end{table}
\begin{table}[h]
\begin{center}
\begin{tabular}{|c|c|c|c|c|c|c|}
\hline
No. & \multicolumn{2}{|c|}{Cut-offs}  & State-space size & Designated state  \\
\hline
$i$ & $C_{l,i}$ & $C_{r,i}$ & $n_i$ & $x_\ell$  \\
\hline
\hline
$1$ & $0$ & $360$ & $65341$ & $(235, 115)$ \\
$2$ & $0$ & $860$ & $371,091$  & $(235, 115)$ \\
$3$ & $0$ & $1360$ & $926,841$  &  $(235, 115)$ \\
$4$ & $0$ & $2360$ & $2,788,341$ & $(235, 115)$ \\
$5$ & $0$ & $3360$ & $5,649,841$ & $(235, 115)$ \\
$6$ & $0$ & $4360$ & $9,511,341$ & $(235, 115)$ \\
\hline
\end{tabular}
\end{center}
\caption{Truncated state-spaces and designated states used in the toggle-switch network. The cut-off values ($C_l$ and $C_r$) for the trapezoidal truncations \eqref{defn_trap_trunc} are provided.}
\label{table_tgs_state_spaces}
\end{table} 
\begin{table}[h!]
	\begin{center} {\small
		\begin{tabular}{|c | c | c | c | c |}
			\hline
			Objective & Sensitive & Sensitivity & Sensitive & Sensitivity \\
			Function & Parameter & Values (PE) & Parameter & Values (PE) \\ \hline \hline
			\multirow{3}{*}{$x_1$} & $\theta_1$ & $3.0677 \cdot 10^{-5}$ & $\theta_4$ & $0.057690$ \\
			& $\theta_2$ & $-5.1177 \cdot 10^{-3}$ & $\theta_5$ & $-0.095117$ \\
			& $\theta_3$ & $-1.1531 \cdot 10^{-4}$ & $\theta_6$ & $6.1345 \cdot 10^{-7}$ \\ \hline \hline
			\multirow{3}{*}{$x_2$} & $\theta_1$ & $-7.5857 \cdot 10^{-3}$ & $\theta_4$ & $-1254.8$ \\
			& $\theta_2$ & $1.2645$ & $\theta_5$ & $23.522$ \\
			& $\theta_3$ & $2.5095$ & $\theta_6$ & $-9.8566 \cdot 10^{-3}$ \\ \hline \hline
		\end{tabular} }
	\end{center}
	\caption{Sensitivity analysis for the toggle-switch network. This table lists the PE-estimated sensitivity values w.r.t. all six parameters $\theta_1,\dots,\theta_6$ for two objective functions $f(x) =x_1$ and $f(x) = x_2$ with expected steady-state values $\E_{\theta}(x_1) \approx 0.015148$ and $\E_{\theta}(x_2)  =496.23$ respectively. The corresponding percentage relation error RE$\%$ \eqref{defn_RE_formula} for the sensitivity estimates is also provided. To perform all these computations the total processing time was $782$ seconds on a single processor and the total memory requirement was 45'050 MB.}
	\label{tbl:TSSens}
\end{table}

\begin{table}
    \begin{center}
        \begin{tabular}{|c|c|c|c|}
            \hline
            No. & Reaction & Propensity \\
            \hline
            1 & $\emptyset \rightarrow \mathcal{S}_1$ & $\lambda_1(x,\theta) = \theta_1$ \\
            2 & $\mathcal{S}_1 \rightarrow \emptyset$ & $\lambda_2(x,\theta) = \theta_2 x_1$ \\
            3 & $\emptyset \rightarrow \mathcal{S}_2$ & $\lambda_3(x,\theta) = \theta_3$ \\
            4 & $\mathcal{S}_2 \rightarrow \emptyset$ & $\lambda_4(x,\theta) = \theta_4 x_2$ \\
            5 & $2\mathcal{S}_2 \rightarrow \mathcal{S}_4$ & $\lambda_5(x,\theta) = \theta_5 x_2 (x_2-1)$ \\
            6 & $\mathcal{S}_4 \rightarrow 2\mathcal{S}_2$ & $\lambda_6(x,\theta) = \theta_6 x_4$ \\
            7 & $\mathcal{S}_1+\mathcal{S}_2 \rightarrow \mathcal{S}_3$ & $\lambda_7(x,\theta) = \theta_7 x_1 x_2$ \\
            8 & $\mathcal{S}_3 \rightarrow \mathcal{S}_1+\mathcal{S}_2$ & $\lambda_8(x,\theta) = \theta_8 x_3$ \\
            \hline
        \end{tabular}
        \caption{Reactions in the deficiency zero network. The propensities are parametrized by the vector $\theta = (\theta_1,\theta_2,\theta_3,\theta_4,\theta_5,\theta_6)$. }
        \label{DefZero}
    \end{center}
\end{table}

\begin{table}[h]
\begin{center}
\begin{tabular}{|c|c|c|c|c|c|c|}
\hline
No. & \multicolumn{2}{|c|}{Cut-offs}  & State-space size & Designated state  \\
\hline
$i$ & $C_{l,i}$ & $C_{r,i}$ & $n_i$ & $x_\ell$  \\
\hline
\hline
$1$ & $0$ & $15$ & $3,876$ & $(10,0,0,0)$ \\
$2$ & $0$ & $20$ & $10,626$  & $(10,0,0,0)$ \\
$3$ & $0$ & $25$ & $23,751$  &  $(10,0,0,0)$ \\
$4$ & $0$ & $30$ & $46,376$ & $(10,0,0,0)$ \\
$5$ & $0$ & $35$ & $82,251$ & $(10,0,0,0)$ \\
$6$ & $0$ & $40$ & $135,751$ & $(10,0,0,0)$ \\
$7$ & $0$ & $45$ & $211,876$ & $(10,0,0,0)$ \\
$8$ & $0$ & $50$ & $316,251$ & $(10,0,0,0)$ \\
$9$ & $0$ & $55$ & $455,126$ & $(10,0,0,0)$ \\
\hline
\end{tabular}
\end{center}
\caption{Truncated state-spaces and designated states used in the deficiency zero network. The cut-off values ($C_l$ and $C_r$) for the trapezoidal truncations \eqref{defn_trap_trunc} are provided.}
\label{table_dz_state_spaces}
\end{table} 

\begin{table}[h!]
    \begin{center} {\small
        \begin{tabular}{|c | c | c | c | c | c | c |}
            \hline
            Objective & Sensitive & Sensitivity & \multirow{2}{*}{RE\%} & Sensitive & Sensitivity & \multirow{2}{*}{RE\%} \\
            Function & Parameter & Values (PE) & & Parameter & Values (PE) & \\ \hline \hline
            \multirow{4}{*}{$x_1$} & $\theta_1$ & $1.2450$ & $8.0956 \cdot 10^{-8}$ & $\theta_5$ & $-3.2680 \cdot 10^{-9}$ &  $3.2680 \cdot 10^{-9}$ \\
            & $\theta_2$ & $-7.0312$ & $8.2053 \cdot 10^{-8}$ & $\theta_6$ & $1.7825 \cdot 10^{-10}$ & $1.7825 \cdot 10^{-10}$ \\
            & $\theta_3$ & $-1.5332  \cdot 10^{-9}$ & $1.5332  \cdot 10^{-9}$ & $\theta_7$ & $-5.3773 \cdot 10^{-10}$ & $5.3773 \cdot 10^{-10}$ \\
            & $\theta_4$ & $7.6662 \cdot 10^{-9}$ & $7.6662 \cdot 10^{-9}$ & $\theta_8$ & $2.0159 \cdot 10^{-11}$ & $2.0159 \cdot 10^{-11}$ \\ \hline \hline
            \multirow{4}{*}{$x_2$} & $\theta_1$ & $-1.4302 \cdot 10^{-9}$ & $1.4302 \cdot 10^{-9}$ & $\theta_5$ & $-7.7320 \cdot 10^{-9}$ & $7.7320 \cdot 10^{-9}$ \\
            & $\theta_2$ & $8.0446 \cdot 10^{-9}$ & $8.0446 \cdot 10^{-9}$ & $\theta_6$ & $4.2174 \cdot 10^{-10}$ & $4.2174 \cdot 10^{-10}$\\
            & $\theta_3$ & $1.0000$ & $3.5816 \cdot 10^{-7}$ & $\theta_7$ & $-9.3908 \cdot 10^{-10}$ & $9.3908 \cdot 10^{-10}$ \\
            & $\theta_4$ & $-5.0000$ &  $3.5429 \cdot 10^{-7}$ & $\theta_8$ & $3.5205 \cdot 10^{-11}$ & $3.5205 \cdot 10^{-10}$ \\ \hline \hline
            \multirow{4}{*}{$x_3$} & $\theta_1$ & $0.23438$ & $1.0319 \cdot 10^{-7}$ & $\theta_5$ & $-7.6189 \cdot 10^{-9}$ & $7.6189 \cdot 10^{-9}$ \\
            & $\theta_2$ & $-1.3184$ & $1.0260 \cdot 10^{-7}$ & $\theta_6$ & $4.1558 \cdot 10^{-10}$ & $4.1558 \cdot 10^{-10}$ \\
            & $\theta_3$ & $0.21094$ & $1.8204 \cdot 10^{-8}$ & $\theta_7$ & $0.35156$ & $6.7294 \cdot 10^{-9}$  \\
            & $\theta_4$ & $-1.0547$ & $1.3920 \cdot 10^{-8}$ & $\theta_8$ & $-0.013184$ & $6.6909 \cdot 10^{-9}$ \\ \hline \hline
            \multirow{4}{*}{$x_4$} & $\theta_1$ & $4.1305 \cdot 10^{-10}$ & $4.1305 \cdot 10^{-10}$ & $\theta_5$ & $2.2727$ & $4.8438 \cdot 10^{-7}$ \\
            & $\theta_2$ & $-2.3228 \cdot 10^{-9}$ & $2.3228 \cdot 10^{-9}$ & $\theta_6$ & $-0.12397$ & $4.8415 \cdot 10^{-7}$ \\
            & $\theta_3$ & $0.54545$ & $9.6426 \cdot 10^{-7}$ & $\theta_7$ & $-5.3900 \cdot 10^{-10}$ & $5.3900 \cdot 10^{-10}$ \\
            & $\theta_4$ & $-2.7273$ & $9.6111 \cdot 10^{-7}$ & $\theta_8$ & $2.0204 \cdot 10^{-11}$ & $2.0204 \cdot 10^{-11}$ \\ \hline \hline
        \end{tabular} }
    \end{center}
    \caption{Sensitivity analysis for the deficiency zero network. This table lists the PE-estimated sensitivity values w.r.t. all eight parameters $\theta_1,\dots,\theta_8$ for four objective functions $f(x) =x_1$, $f(x) = x_2$, $f(x) =x_3$ and $f(x)=x_4$ with expected steady-state values $\E_{\theta}(x_1) = 5.625, \E_{\theta}(x_2) =  5, \E_{\theta}(x_3) = 1.0547$ and $\E_{\theta}(x_4)  = 1.3636$ respectively. The corresponding percentage relation error RE$\%$ \eqref{defn_RE_formula} for the sensitivity estimates is also provided. To perform all these computations the total processing time was $4386$ seconds on a single processor and the total memory requirement was 26'458 MB.}
    \label{tbl:defZeroSens}
\end{table}

\begin{table}
    \begin{center}
        \begin{tabular}{|c|c|c|c|}
            \hline
            No. & Reaction & Propensity \\
            \hline
            1 & $\mathcal{S}_1 \rightarrow \mathcal{S}_2$ & $\lambda_1(x,\theta) = \theta_1 x_1$ \\
            2 & $\mathcal{S}_2 \rightarrow \mathcal{S}_1$ & $\lambda_2(x,\theta) = \theta_2 x_2$ \\
            3 & $\mathcal{S}_2 \rightarrow \mathcal{S}_3$ & $\lambda_3(x,\theta) = \theta_3 x_2$ \\
            4 & $\mathcal{S}_3 \rightarrow \mathcal{S}_2$ & $\lambda_4(x,\theta) = \theta_4 x_3$ \\
            \hline
        \end{tabular}
        \caption{Reactions in the simple linear network. The propensities are parametrized by the vector $\theta = (\theta_1,\theta_2,\theta_3,\theta_4)$. }
        \label{PaperSystem}
    \end{center}
\end{table}

\begin{table}[h!]
\begin{center} {\small
        \begin{tabular}{|c | c | c | c | c | c | c |}
            \hline
            Objective & Sensitive & Sensitivity & \multirow{2}{*}{RE\%} & Sensitive & Sensitivity & \multirow{2}{*}{RE\%} \\
            Function & Parameter & Values (PE) & & Parameter & Values (PE) & \\ \hline \hline
            \multirow{2}{*}{$x_1$} & $\theta_1$ & $-0.24691$ & $1.0137 \cdot 10^{-5}$ & $\theta_3$ & $-49.383$ & $1.8538 \cdot 10^{-5}$ \\
            & $\theta_2$ & $0.12346$ & $1.0242 \cdot 10^{-5}$ & $\theta_4$ & $74.074$ & $9.3173 \cdot 10^{-6}$ \\ \hline \hline
            \multirow{2}{*}{$x_2$} & $\theta_1$ & $0.098765$ & $8.0763 \cdot 10^{-6}$ & $\theta_3$ & $-24.691$ & $1.8642 \cdot 10^{-5}$ \\
            & $\theta_2$ & $-0.049383$ & $8.1807 \cdot 10^{-6}$ & $\theta_4$ & $37.037$ & $9.2132 \cdot 10^{-6}$ \\ \hline \hline
            \multirow{2}{*}{$x_3$} & $\theta_1$ & $0.14815$ & $1.1511 \cdot 10^{-5}$ & $\theta_3$ & $74.074$ & $1.8572 \cdot 10^{-5}$ \\
            & $\theta_2$ & $-0.074074$ & $1.1616 \cdot 10^{-5}$ & $\theta_4$ & $-111.11$ & $9.2826 \cdot 10^{-6}$ \\ \hline \hline
        \end{tabular} }
    \end{center}
    \caption{Sensitivity analysis for the simple linear network. This table lists the PE-estimated sensitivity values w.r.t. all four parameters $\theta_1,\dots,\theta_4$ for three objective functions $f(x) =x_1$, $f(x) = x_2$ and $f(x)=x_3$ with expected steady-state values $\E_{\theta}(x_1) = 4.4444, \E_{\theta}(x_2) =  2.2222$ and $\E_{\theta}(x_3)  = 3.3333$ respectively. The corresponding percentage relation error RE$\%$ \eqref{defn_RE_formula} for the sensitivity estimates is also provided. To perform all these computations the total processing time was only $0.2$ seconds on a single processor and the total memory requirement was 105 MB.}
    \label{tbl:PaperSystemSens}
\end{table}

%\section*{Acknowledgments}

\bibliographystyle{unsrt}
%\bibliography{references}

\end{document}